\documentclass[11pt]{article}
\usepackage{amsmath}
\usepackage{amsthm}
\usepackage{mathtools}
\usepackage{amssymb}
\usepackage{geometry}
\usepackage{cases}
\usepackage{color}
\usepackage{authblk}
\usepackage{listings}
\usepackage{float}
\usepackage{hyperref}
\usepackage{amsmath,amsfonts,amssymb,amsthm,mathtools}
\usepackage{latexsym,graphicx}
\usepackage{dsfont}
\usepackage{comment}
\usepackage{amscd}
\usepackage[dvipsnames]{xcolor}
\usepackage{extarrows} 
	\numberwithin{equation}{section} 
\usepackage{yfonts}
\usepackage{graphicx}
\usepackage{appendix}
\usepackage{caption}
\usepackage{subcaption}
\usepackage{verbatim}
\usepackage{authblk}
\usepackage{doi}

\newcommand{\ee}{\mathrm{e}}
\newcommand{\ii}{\mathrm{i}}

\newcommand{\T}{\mathbb{T}}
\newcommand{\C}{\mathbb{C}}
\newcommand{\R}{\mathbb{R}}
\newcommand{\D}{\mathbb{D}}
\newcommand{\Z}{\mathbb{Z}}

\usepackage{xcolor}
\definecolor{riverlane_green}{RGB}{0, 111, 98}
\definecolor{riverlane_light_green}{RGB}{0, 150, 143}
\definecolor{riverlane_orange}{RGB}{255, 117, 0}
\definecolor{riverlane_red}{RGB}{220, 68, 5}
\definecolor{riverlane_pink}{RGB}{207, 111, 127}
\usepackage{hyperref}
\hypersetup{
  colorlinks   = true, 
  urlcolor     = riverlane_green, 
  linkcolor    = riverlane_orange, 
  citecolor   = riverlane_green  
}

\def\Xint#1{\mathchoice
{\XXint\displaystyle\textstyle{#1}}%
{\XXint\textstyle\scriptstyle{#1}}%
{\XXint\scriptstyle\scriptscriptstyle{#1}}%
{\XXint\scriptscriptstyle\scriptscriptstyle{#1}}%
\!\int}
\def\XXint#1#2#3{{\setbox0=\hbox{$#1{#2#3}{\int}$}
\vcenter{\hbox{$#2#3$}}\kern-.5\wd0}}

\def\pvint{\Xint-}

\newtheorem{prob}{Problem}
\newtheorem{proposition}{Proposition}[subsection]
\newtheorem{theorem}{Theorem}
\newtheorem{corollary}{Corollary}[theorem]
\newtheorem{lemma}{Lemma}[subsection]

\newtheoremstyle{break}
  {\topsep}{\topsep}%
  {\itshape}{}%
  {\bfseries}{}%
  {\newline}{}%
\theoremstyle{break}
\newtheorem{algo}{Algorithm}

\theoremstyle{}
\newtheorem{rem}{Remark}[subsection]

\newcommand\blfootnote[1]{%
  \begingroup
  \renewcommand\thefootnote{}\footnote{#1}%
  \addtocounter{footnote}{-1}%
  \endgroup
}

\title{Complementary polynomials in quantum signal processing}

\begin{document}

\author[1,2]{Bjorn K. Berntson}
\author[2]{Christoph S\"{u}nderhauf}
\affil[1]{Riverlane Research,
Cambridge, Massachusetts}
\affil[2]{Riverlane,
Cambridge, United Kingdom}

\maketitle

\begin{abstract}
Quantum signal processing is a framework for implementing polynomial functions on quantum computers. To implement a given polynomial $P$, one must first construct a corresponding \textit{complementary polynomial} $Q$. Existing approaches to this problem employ numerical methods that are not amenable to explicit error analysis. We present a new approach to complementary polynomials using complex analysis. Our main mathematical result is a contour integral representation for a canonical complementary polynomial. On the unit circle, this representation has a particularly simple and efficacious Fourier analytic interpretation, which we use to develop a Fast Fourier Transform-based algorithm for the efficient calculation of $Q$ in the monomial basis with explicit error guarantees. Numerical evidence that our algorithm outperforms the state-of-the-art optimization-based method for computing complementary polynomials is provided. 
\end{abstract}

\section{Introduction}
\blfootnote{Emails: bjorn.berntson@riverlane.com, christoph.sunderhauf@riverlane.com}Quantum signal processing (QSP) \cite{low2017} and its extensions \cite{wang2023, motlagh2024} describe simple single-qubit parameterized circuits that apply a chosen polynomial to a scalar. They have become indispensable in state-of-the-art quantum algorithms because the circuits can be lifted from scalars to arbitrary matrices, resulting in the quantum singular value transformation (QSVT) \cite{gilyen2019,sunderhauf2023}. The matrix is embedded inside a larger unitary, called a block encoding. The polynomial is then applied to all singular values of the matrix simultaneously, which can lead to a quantum speedup compared to classical evaluation of matrix functions. The QSVT has had a tremendous impact in quantum computing since polynomials may be used to approximate a wide variety of functions. Thereby, this family of algorithms encompasses many prior quantum algorithms \cite{martyn2021}, including those for Hamiltonian simulation \cite{berry2024}, solving linear systems \cite{harrow2009}, phase estimation \cite{Rall_2021}, and amplitude amplification \cite{brassard2002}, often even improving the prior algorithm. \\

There are different parameterizations of QSP \cite{low2017,gilyen2019,haah2019}. The fundamental idea underlying each is that certain polynomials may be realized within a matrix element of a finite product of single-qubit unitaries \cite{low2017, motlagh2024,alexis2023}. Determining the parameters of these unitaries is the obstruction to implementing particular polynomials in QSP. 
To delineate the mathematical structure of this factorization problem, we recall the generalized quantum signal processing (GQSP) framework\footnote{For future convenience, in particular to relate \eqref{eq:GQSP} to other QSP conventions in Appendix~\ref{app:QSP}, we have written \eqref{eq:GQSP} as the transpose of \cite[Eq.~7]{motlagh2024}.} of Motlagh and Wiebe \cite{motlagh2024}, which, as we show in Appendix~\ref{app:QSP}, subsumes standard formulations of QSP.

\begin{theorem}[Generalized quantum signal processing, \cite{motlagh2024}]
\label{thm:GQSP}
Let $P\in \C[z]$ such that $\deg P=d\in\Z_{\geq 1}$ and $\lvert P(z)\rvert\leq 1$ on $\T\coloneqq \{z\in \C:\lvert z\rvert=1\}$. Then, there exists $Q\in \C[z]$  such that $\deg Q=d$ and
\begin{equation}\label{eq:PQ}
\lvert P(z)\rvert^2+\lvert Q(z)\rvert^2=1 \quad (z\in \T)
\end{equation}
holds. Moreover, there exist parameters $\lambda\in (-\pi,\pi]$ and $(\theta_j)_{j=0}^d,(\phi_j)_{j=0}^d\in (-\pi,\pi]^{d+1}$ such that
\begin{equation}\label{eq:GQSP}
\left(\begin{array}{cc}
P(z) & Q(z) \\
* & * 
\end{array}\right)=
\left(\begin{array}{cc}
\ee^{\ii(\lambda+\phi_0)}\cos \theta_0 & \ee^{\ii\lambda}\sin \theta_0 \\
\ee^{\ii\phi_0}\sin \theta_0 & -\cos \theta_0
\end{array}\right)\Bigg[\prod_{j=1}^d \left(\begin{array}{cc} z & 0 \\ 0 & 1 \end{array}\right)\left(\begin{array}{cc}
\ee^{\ii\phi_j}\cos \theta_j & \sin \theta_j \\
\ee^{\ii\phi_j}\sin \theta_j & -\cos \theta_j
\end{array}\right) \Bigg] \quad (z\in \T),
\end{equation}
where $*$ indicates the precise form of the matrix elements is immaterial, holds. 
\end{theorem}

In \eqref{eq:GQSP}, we call $Q$ a \textit{complementary polynomial} to $P$ and the parameters $\lambda$, $(\phi_j)_{j=0}^d$, $(\theta_j)_{j=0}^d$ \textit{phase factors}. Given a complementary polynomial, the phase factors may be constructed via an exact iterative method \cite[Algorithm~1]{motlagh2024}; analogous statements hold for standard QSP, see \cite[Theorem~3]{gilyen2019} and \cite[Section~3.1.2]{ying2022}. Importantly, the phase factors required in lifted algorithms operating on matrices like the QSVT follow immediately from those obtained in QSP. Motivated by the above discussion, we introduce the following problem, whose theoretical and numerical resolution is the subject of this paper.

\begin{prob}[Complementary polynomials problem]\label{prob:CP}
Given $P\in \C[z]$ satisfying the conditions of Theorem~\ref{thm:GQSP}, find $Q\in \C[z]$ in the monomial basis, such that $\deg Q=\deg P$ and \eqref{eq:PQ} holds. 
\end{prob}

Previous approaches to this problem rely on root-finding \cite{gilyen2019}, Prony's method \cite{ying2022}, or optimization \cite{motlagh2024}; see Section~\ref{subsec:related}. Here, we construct an exact representation of a canonical complementary polynomial $Q$, valid throughout the entire complex plane, in the form of a set of contour integrals; the problem of constructing the complementary polynomial is thus reduced to quadratures. The contour integral representation of $Q$ on $\T$ can be rephrased in the language of Fourier analysis. We combine this Fourier analytic interpretation and the Fast Fourier Transform (FFT) to develop efficient numerical algorithms to compute $Q$ in the monomial basis. \\ 

Exact and explicit error analysis of our algorithms is performed, providing rigorous upper bounds on the classical runtimes. Furthermore, numerical results from our reference implementation demonstrate the practicality and competitiveness of our algorithm. We emphasize that existing numerical approaches to the construction of complementary polynomials, which we describe in Section~\ref{subsec:related} below, rely on heuristics and so are not amenable to rigorous error analysis. \\

In the remainder of this introduction, we state our results, give remarks on their proofs, describe related literature, introduce notation used in the main text, and outline the structure of the paper. 

\subsection{Statement of results}
We first construct a set of contour integral representations for $Q$. Let
\begin{equation}\label{eq:P}
P(z)=\sum_{n=0}^d p_nz^n \quad (p_0\neq 0);
\end{equation}
the restriction that $p_0\neq 0$ is imposed without loss of generality as $\lvert z^n P(z)\rvert=\lvert P(z)\rvert$ holds for all $z\in \T$, $n\in \Z$. Our main result is the following theorem, giving representations of $Q$ on $\D\coloneqq \{z\in \C:\lvert z\rvert<1\}$, $\T$, and $\C\setminus\overline{\D}$. 

\begin{theorem}[Contour integral representation of the canonical complementary polynomial]
\label{thm:main}
Suppose $P$ satisfying the assumptions of Problem~\ref{prob:CP} is given in the form \eqref{eq:P}. Let $d_0\in \Z_{\geq 0}$ be the number of roots of $1-\lvert P(z)\rvert^2$ on $\T$, not counting multiplicity, and $\{(t_j,2\alpha_j)\}_{j=1}^{d_0}$ be the corresponding roots and multiplicities, which are necessarily even. Then, 
\begin{subnumcases}{\label{eq:Q} Q(z)=} 
{Q}_0(z)\exp\Bigg(\frac{1}{4\pi\ii}\int_{\T} \frac{z'+z}{z'-z}\log \bigg(\frac{1-\lvert P(z')\rvert^2}{\lvert {Q}_0(z')\rvert^2}\bigg)	\frac{\mathrm{d}z'}{z'}\Bigg) & $z\in \D$ \label{eq:QD} \\
{Q}_0(z)\exp\Bigg(\frac{1}{4\pi\ii}\,\pvint_{\T}\frac{z'+z}{z'-z}\log \bigg(\frac{1-\lvert P(z')\rvert^2}{\lvert {Q}_0(z')\rvert^2}\bigg) \frac{\mathrm{d}z'}{z'}+\frac12\log \bigg(\frac{1-\lvert P(z)\rvert^2}{\lvert {Q}_0(z)\rvert^2}\bigg) \Bigg) & $z\in \T$ \label{eq:QT} \\
\frac{1-P(z)P^*(1/z)}{Q_0^*(1/z)}\exp\Bigg(\frac{1}{4\pi\ii}\int_{\T}\frac{z'+z}{z'-z}\log \bigg(\frac{1-\lvert P(z')\rvert^2}{\lvert {Q}_0(z')\rvert^2}\bigg) \frac{\mathrm{d}z'}{z'}\Bigg)	 & $z\in \C\setminus \overline{\D}$,\label{eq:QCD}
\end{subnumcases}
where
\begin{equation}\label{eq:Qtilde}
{Q}_0(z)\coloneqq \prod_{j=1}^{d_0} (z-t_j)^{\alpha_j},
\end{equation}
the integration contour $\T$ is positively-oriented, and the dashed integral indicates a Cauchy principal value prescription \eqref{eq:PV} with respect to the singularity $z'=z$ on $\T$, solves Problem~\ref{prob:CP}. Moreover, \eqref{eq:Q} is, up to a multiplicative phase, the unique solution of Problem~\ref{prob:CP} with no roots in $\D$. 
\end{theorem}

Theorem~\ref{thm:main} provides an exact representation of $Q$ in Problem~\ref{prob:CP}, in 
\textit{canonical form}: all roots lie outside of $\D$. In the important special case where $P$ has real coefficients, a canonical complementary polynomial with real coefficients exists. The precise statement, a mild generalization of a result in \cite[Section~8]{chao2020}, is now given.

\begin{corollary}[Real complementary polynomials]
\label{cor:real}
Let $P\in \R[z]$ satisfying the conditions of Problem~\ref{prob:CP} be given. Then, the canonical complementary polynomial fulfills $Q\in \R[z]$, up to a multiplicative phase. 
\end{corollary}

To obtain $Q$ explicitly in the monomial basis, i.e., to solve Problem~\ref{prob:CP}, it suffices to evaluate \eqref{eq:Q} at any $d+1$ distinct points of $\C$ and employ the Lagrange interpolation formula. Our numerical approach is based on interpolation through roots of unity; this is equivalent to a discrete Fourier transform. The following corollary of Theorem~\ref{thm:main} establishes a Fourier analytic variant of the integral representation \eqref{eq:QT} of $Q$ on $\T$, which we will later use to evaluate $Q$ at the roots of unity. 

\begin{corollary}[Fourier analytic variant of Theorem~\ref{thm:main} on $\T$]
\label{cor:main}
The representation \eqref{eq:QT} of $Q$ on $\T$ is equivalent to 
\begin{equation}\label{eq:QFourier}
Q(\ee^{\ii\theta})={Q}_0(\ee^{\ii\theta})\exp\Bigg( \Pi \bigg[\log\bigg(\frac{1-\lvert P(\ee^{\ii\theta})\rvert^2}{\lvert{Q}_0(\ee^{\ii\theta})\rvert^2}\bigg)\bigg] \Bigg) \quad (\theta\in (-\pi,\pi]),
\end{equation}
where $\Pi$ is the Fourier multiplier defined by
\begin{equation}\label{eq:Piexp}
\Pi[\ee^{\ii n \theta}]\coloneqq \begin{cases}
\ee^{\ii n\theta} & n\in \Z_{> 0} \\
\frac12 & n=0 \\
0 & n\in \Z_{<0}.	
\end{cases}	
\end{equation}
\end{corollary}

Due to \eqref{eq:Piexp}, we have
\begin{equation}
\label{eq:Pi description}
\Pi\Bigg[\sum_{n\in \Z} a_n \ee^{\ii n\theta}\Bigg]=\frac12 a_0+\sum_{n=1}^{\infty} a_n \ee^{\ii n \theta},
\end{equation}
and Corollary~\ref{cor:main} shows that, essentially, constructing $Q$ on $\T$ consists in evaluating the Fourier coefficients of the function $\log\big(\frac{1-\lvert P(\ee^{\ii\theta})\rvert^2}{\lvert {Q}_0(\ee^{\ii\theta})\rvert^2}\big)$. \\

\paragraph{Numerical methods.}
Corollary~\ref{cor:main} suggests a practical numerical method to compute an approximation of the complementary polynomial, supposing $Q_0$ is known. This is trivially the case if 
\begin{equation}
\label{eq:deltabound}
\lVert P(z)\rVert_{\infty,\T}\coloneqq \max_{z\in \T} \,\lvert P(z)\rvert  \leq 1-\delta \quad (\delta\in (0,1));
\end{equation}
then, $Q_0(z)=1$. If $\lVert P(z)\rVert_{\infty,\T}\leq 1$ is guaranteed but a tighter bound \eqref{eq:deltabound} either does not exist or is unknown, $Q_0(z)=1$ can be attained by slightly rescaling $P(z) \to (1-\delta)P(z)$ for a suitable $\delta\in (0,1)$.\\

In cases with a known bound \eqref{eq:deltabound}, Algorithm~\ref{algo:main} solves Problem~\ref{prob:CP} in time
$O(N\log N)$ using a sequence of FFTs, where the even parameter $N\in \Z_{\geq d}$ defines the discrete Fourier basis $\{\ee^{\ii n\theta}\}_{n=-\frac{N}{2}+1}^{\frac{N}{2}}$. An informal description of our algorithm, based on Corollary~\ref{cor:main} with $Q_0(z)=1$, is as follows.
\begin{enumerate}
\item Compute approximations to the Fourier coefficients $(a_n)_{n=-\frac{N}{2}+1}^{\frac{N}2}$ of $\log\big(1-\lvert P(\ee^{\ii\theta})\rvert^2\big)$ using an FFT, in time $O(N\log N)$.  
\item Compute approximations to $Q$ at the $N$th roots of unity by applying the Fourier multiplier $\Pi$ in Fourier space \eqref{eq:QFourier}--\eqref{eq:Pi description}, using an FFT and inverse FFT in time $O(N\log N)$.
\item Compute approximations to the coefficients of $Q/{Q}_0$ in the monomial basis using the result of the previous step and an FFT, in time $O(N\log N)$.
\end{enumerate}
We prove in Theorem~\ref{thm:error} that this algorithm is efficient, with a sufficient $N$ scaling as $N=O\!\left(\frac{d}{\delta}\log\frac{d}{\delta\varepsilon}\right)$, where $\varepsilon$ is the error in the in the monomial basis coefficients. 
\\

In the case where $\lVert P(z)\rVert_{\infty,\T}\leq 1$, but a tighter upper bound \eqref{eq:deltabound} either (i) does not exist, (ii) is unknown, or (iii) has $\delta$ so small that the upper bound on the runtime of Algorithm~\ref{algo:main} is undesirable, we recourse to Algorithm~\ref{algo:downscaling}. In this algorithm, the input polynomial is appropriately downscaled and then input into Algorithm~\ref{algo:main}. Theorem~\ref{thm:error downscaling} proves that a sufficient $N=O\!\left(\frac{d}{\varepsilon}\log\frac{d}{\varepsilon}\right)$, where $\varepsilon$ is the error in the complementarity condition \eqref{eq:PQ}; observe that $N$ is independent of $\delta$.\\

 We give a reference implementation of Algorithm~\ref{algo:main} and provide numerical evidence that it outperforms the optimization-based approach to complementary polynomials from \cite{motlagh2024} in Section~\ref{sec:numericalresults}.

\subsection{Remarks on the results and their proofs}

The following remarks apply to Theorem~\ref{thm:main} and Corollary~\ref{cor:main}. Below, we reference classical complex analysis theorems; precise statements of these theorems may be found in Appendix~\ref{app:complex}.

\begin{enumerate}

\item The crucial observation leading to Theorem~\ref{thm:main} is that the real part of the function $\log(Q(z)/Q_0(z))$, where $Q$ is chosen so that all roots lie outside $\overline{\D}$ and the branch cuts are chosen appropriately, can be determined exactly on $\T$ using \eqref{eq:PQ}. Then, the Schwarz integral formula \cite{ahlfors} is used to construct $Q$ on $\D$. 

\item The representation \eqref{eq:Q} of $Q$ is not manifestly a polynomial. Rather, as elaborated in Section~\ref{sec:proof}, it follows from (i) the existence of a canonical solution of Problem~\ref{prob:CP} via the F\'{e}jer-Riesz theorem \cite{riesz2012} and (ii) the uniqueness of holomorphic functions constructed by the Schwarz integral formula that $Q$ is a polynomial. 

\item We show within the proof of Theorem~\ref{thm:main} that, up to a multiplicative phase, the number of distinct solutions of Problem~\ref{prob:CP} is equal to $\prod_{j=1}^{d_1}(\beta_j+1)$, where $\beta_j$ is the multiplicity of the $j$th root ($j\in [d_1]$) of $Q$ outside of $\overline{\D}$. However, constructing all of these solutions requires knowledge of all roots of $1-P(z)P^*(1/z)$ on $\C$. We construct a canonical solution of Problem~\ref{prob:CP} with no roots on $\D$ in \eqref{eq:Q}; this requires only the knowledge of the roots of $1-\lvert P(z)\rvert^2$ on $\T$. 
\item Theorem~\ref{thm:main} and Corollary~\ref{cor:main} can be proven in different ways. Consider the function $\log(Q(z)/Q_0(z))$, where $Q$ is chosen so that all roots lie outside $\overline{\D}$ and the branch cuts are chosen appropriately. One can use \eqref{eq:PQ} and the Fej\'{e}r-Riesz theorem to construct a scalar Riemann-Hilbert problem \cite{ablowitz2003} on $\T$ for $\log(Q(z)/Q_0(z))$; this Riemann-Hilbert problem is explicitly solvable using a Cauchy integral, from which \eqref{eq:QD} follows. Corollary~\ref{cor:main} can be proven directly using \eqref{eq:PQ}, the Fej\'{e}r-Riesz theorem, and the fact that the periodic Hilbert transform \eqref{eq:H} relates the real and imaginary parts of the boundary values of a function holomorphic on $\D$, namely $\log(Q(z)/Q_0(z))$.
\end{enumerate}

\subsection{Related work}
\label{subsec:related}

 In the QSP literature, a variety of numerical methods for solving Problem~\ref{prob:CP} or its avatars have been developed. As is evident from the proof of Theorem~\ref{thm:main} in Section~\ref{subsec:mainproof}, knowledge of all roots of $1-P(z)P^*(1/z)$ allows for the explicit construction of $Q$; analogous statements hold for complementary polynomials in standard QSP. Thus, employing standard root-finding algorithms provides a straightforward means to calculate complementary polynomials \cite{gilyen2019,haah2019,chao2020}. Root-finding algorithms are known to be expensive and suffer from numerical instability; the highest-degree polynomial successfully treated with this approach was reported to have degree $d=3
 \times 10^3$ \cite{chao2020}. An alternative method that avoids root-finding and instead directly calculates the characteristic polynomial of the roots of $1-P(z)P^*(1/z)$ within $\D$ using Prony's method has been proposed in \cite{ying2022}. Numerical experiments have demonstrated the effectiveness of this approach for polynomials with degree up to $d=5\times 10^4$.\\
 
 The current state-of-the-art method for Problem~\ref{prob:CP} was developed in \cite{motlagh2024}. There, a loss function derived from the complementarity condition \eqref{eq:PQ} is minimized to determine $Q$ with an optimization procedure; this approach was demonstrated to be effective for $d$ up to the order of $10^7$, achieving accuracies as low as $10^{-6}$ in the loss function.  
 In this paper, we present numerical results showing that Algorithm~\ref{algo:main} is effective for the same degrees, up to $d=10^7$. At the same time, our algorithm requires much shorter runtimes and achieves better accuracies, even without the GPU acceleration used in \cite{motlagh2024}. \\

  Given a complementary polynomial, it remains to calculate the phase factors. In both GQSP \cite{motlagh2024} and standard QSP \cite{gilyen2019, ying2022}, exact recursive formulas may be used to determine phase factors. Variations on and improvements to this approach are described in \cite{haah2019} and \cite{chao2020}. We also mention that, as an alternative to the approaches described above, optimization-based methods to compute phase factors without knowledge of the complementary polynomial have been developed in \cite{dong2021,wang2021,dong2023}. These methods have been used to determine phase factors for polynomials up to degree $d=10^5$.

\subsection{Notation}

We write the complex conjugate of $z\in \C$ as $z^*$ and for a function $f:\C\to \C$, define $f^*(z)\coloneqq f(z^*)^*$. For a set $X\subset \C$, we write $\partial X$ and $\overline{X}$ for its boundary and closure, respectively. We define $\D$ and $\T$ to be the open unit disk and unit circle in the complex plane, respectively. Given an integer $N\in \Z_{\geq 1}$, we define the sets $[N]\coloneqq \{n\in \Z_{\geq 1}: n\leq  N\}$ and $[N]_0\coloneqq [N]\cup\{0\}$. Dashed integrals indicate a Cauchy principal value prescription with respect to singularities of the integrand on the integration contour. Contour integrals are always assumed to carry a positive orientation. Unless otherwise indicated, all logarithms are with respect to base $\ee$. We denote by $\lVert \cdot \rVert_{\infty,X} $ the uniform norm on $X$.

\subsection{Plan of the paper}

Theorem~\ref{thm:main} and Corollaries~\ref{cor:real} and \ref{cor:main} are proved in Section~\ref{sec:proof}. In Section~\ref{sec:numerical}, we design two algorithms for the computation of $Q$ following from Corollary~\ref{cor:main}. Error analysis of our algorithms is performed in Section~\ref{sec:error} and numerical results comparing our algorithm to the optimization-based approach of \cite{motlagh2024} are presented in Section~\ref{sec:numericalresults}. Section~\ref{sec:discussion} contains a discussion of our results and possibilities for future work. In Appendix~\ref{app:QSP}, we show that different QSP parameterizations can be viewed as special cases of GQSP. Appendix~\ref{app:complex} contains precise statements of the complex analysis theorems used to prove our results.

\section{Proofs of main results}
\label{sec:proof}

We provide rigorous proofs of the mathematical results reported in the previous section. Theorem~\ref{thm:main} is proved in Section~\ref{subsec:mainproof} and Corollaries~\ref{cor:real} and \ref{cor:main} are proved in Sections~\ref{subsec:realproof} and \ref{subsec:corproof}, respectively. 

\subsection{Proof of Theorem~\ref{thm:main}}
\label{subsec:mainproof}
Observe that on $\T$, $1-\lvert P(z)\rvert^2=1-P(z)P^*(1/z)$, a positive-semidefinite Laurent polynomial of degree $d$. Thus, by the Fej\'{e}r-Riesz theorem, there exists $Q\in \C[z]$ so that $\deg Q=d$, $Q$ is nonzero on $\D$, any root of $Q$ on $\T$ has even multiplicity, and \eqref{eq:PQ} is satisfied. We write
\begin{equation}\label{eq:QFR}
Q(z)=\bar{Q}\Bigg(\prod_{j=1}^{d_0}(z-t_j)^{\alpha_j}\Bigg)\Bigg(\prod_{j=1}^{d_1}(z-w_j)^{\beta_j}\Bigg) \quad (z\in \C),
\end{equation}
where $\bar{Q}\in \C\setminus\{0\}$ and $
\{(w_j,\beta_j)
\}_{j=1}^{d_1}$ are the roots of $Q$ outside of $\D$ with corresponding multiplicities; recall that $2\alpha_j$ is the multiplicity of the root $t_j$. It follows from \eqref{eq:PQ} and \eqref{eq:QFR} that
\begin{equation}\label{eq:PQ2}
1-\lvert P(z)\rvert^2= \lvert Q(z)\rvert^2=\lvert \bar{Q} \rvert^2\Bigg( \prod_{j=1}^{d_0} (z-t_j)^{\alpha_j}\bigg(\frac{1}{z}-\frac{1}{t_j}\bigg)^{\alpha_j}\Bigg)\Bigg( \prod_{j=1}^{d_1} (z-w_j)^{\beta_j}\bigg(\frac{1}{z}-w_j^*\bigg)^{\beta_j}\Bigg) \quad (z\in \T);
\end{equation}
note that this factorization is unique up to rotations $\bar{Q}\to t \bar{Q}$, $t\in \T \simeq \mathrm{U}(1)$. Moreover, we see that transforming
\begin{equation}\label{eq:Qswap}
Q(z)\to \bigg(\frac{1-zw_j^*}{z-w_j}\bigg)^k Q(z) 	\end{equation}
for any $j\in [d_1]$ and $k\in [\beta_j]_0$, preserves \eqref{eq:PQ2} via \eqref{eq:PQ}. It follows that there are $\prod_{j=1}^{d_1}(\beta_j+1)$ distinct solutions of Problem~\ref{prob:CP}, up to $\mathrm{U}(1)$ equivalence. \\

To construct a canonical solution \eqref{eq:QFR} of Problem~\ref{prob:CP}, we combine \eqref{eq:Qtilde} and \eqref{eq:QFR} and write
\begin{equation}
\frac{Q(z)}{{Q}_0(z)}=\bar{Q}\prod_{j=1}^{d_1}(z-w_j)^{\beta_j}.	
\end{equation}
Observe that any logarithm of $Q/Q_0$ will have branch points $\{w_j\}_{j=1}^{d_1}$. Consider the function
\begin{equation}\label{eq:U}
U(z)\coloneqq \log\bigg(\frac{Q(z)}{{Q}_0(z)}\bigg) \quad (z\in \C\setminus B),
\end{equation}
where the branch cuts are chosen to be
\begin{equation}
B=\bigcup_{n=1}^{d_1} \{s w_n: s \in [1,\infty)\}.	
\end{equation}
By construction, $U(z)$ is holomorphic on $\overline{\D}$.  The real part of $U(z)$ is found to be
\begin{equation}\label{eq:ReU}
\mathrm{Re}\,U(z)=\log\bigg\lvert \frac{Q(z)}{Q_0(z)}\bigg\rvert =\frac12\log\bigg(\frac{1-\lvert P(z)\rvert^2}{\lvert Q_0(z) \rvert^2}\bigg) \quad (z\in \C\setminus B).
\end{equation}
In particular, \eqref{eq:ReU} holds on $\T$, so by the Schwarz integral formula \cite{ahlfors}, we obtain a representation of $U(z)$ on $\D$,
\begin{align}\label{eq:USchwarz}
U(z)=&\; \frac{1}{2\pi\ii}\int_{\T}\frac{z'+z}{z'-z}\,\mathrm{Re}\,U(z')\,\frac{\mathrm{d}z'}{z'}+\ii\,\mathrm{Im}\,U(0) \nonumber \\
=&\;  \frac{1}{4\pi\ii}\int_{\T}\frac{z'+z}{z'-z}\log\bigg(\frac{1-\lvert P(z')\rvert^2}{\lvert Q_0(z')\rvert^2}\bigg)\,\frac{\mathrm{d}z'}{z'}+\ii\,\mathrm{Im}\,U(0) \quad (z\in \D). 
\end{align}
By exponentiating \eqref{eq:U} and \eqref{eq:USchwarz} and using the $\mathrm{U}(1)$ symmetry of Problem~\ref{prob:CP}, we obtain \eqref{eq:QD}. \\

The second case of \eqref{eq:Q} is obtained from the first using the Plemelj formula \eqref{eq:Plemelj}. Note that the integrand in \eqref{eq:QD} has a simple pole at $z'=z$ with residue $2\log\big(\frac{1-\lvert P(z)\rvert^2}{\lvert Q_0(z)\rvert^2}\big)$. Thus, applying the Plemelj formula as $z\in \D$ approaches the contour $\T$ gives \eqref{eq:QT}. \\

 The third case of \eqref{eq:Q} is obtained from the second by analytic continuation. Let us write \eqref{eq:QT} as
\begin{align}\label{eq:QCDcalc}
Q(z)=&\; Q_0(z)\exp\Bigg(\frac{1}{4\pi\ii}\,\pvint_{\T}\frac{z'+z}{z'-z}\log \bigg(\frac{1-\lvert P(z')\rvert^2}{\lvert {Q}_0(z')\rvert^2}\bigg) \frac{\mathrm{d}z'}{z'}-\frac12\log \bigg(\frac{1-\lvert P(z)\rvert^2}{\lvert {Q}_0(z)\rvert^2}\bigg)+ \log \bigg(\frac{1-\lvert P(z)\rvert^2}{\lvert {Q}_0(z)\rvert^2}\bigg)\Bigg) \nonumber \\
=&\; Q_0(z)\frac{1-\lvert P(z)\rvert^2}{\lvert Q_0(z)\rvert^2}\exp \Bigg(\frac{1}{4\pi\ii}\,\pvint_{\T}\frac{z'+z}{z'-z}\log \bigg(\frac{1-\lvert P(z')\rvert^2}{\lvert {Q}_0(z')\rvert^2}\bigg) \frac{\mathrm{d}z'}{z'}-\frac12\log \bigg(\frac{1-\lvert P(z)\rvert^2}{\lvert {Q}_0(z)\rvert^2}\bigg)\Bigg) \quad (z\in \T).
\end{align}
Analytic continuation of the prefactor and exponent in \eqref{eq:QCDcalc}, using that the latter represents the boundary values of a Cauchy integral,  to $\C\setminus\overline{\D}$ gives \eqref{eq:QCD}.

\subsection{Proof of Corollary~\ref{cor:real}}
\label{subsec:realproof}

This result follows from the F\'{e}jer-Riesz theorem and properties of Laurent polynomials with real coefficients. \\

On $\T$, $1-\lvert P(z)\rvert^2=1-P(z)P(1/z)$, a positive-semidefinite Laurent polynomial of degree $d$ with real coefficients. By the same argument as in the proof of Theorem~\ref{thm:main}, we may write $Q$ in the canonical form \eqref{eq:QFR}. Because $1-P(z)P(1/z)$ has real coefficients, it has the following symmetries (i) if $w\in \C\setminus(\R\cup \T)$ is a root, so are $1/w$, $w^*$, and $1/w^*$ and (ii) if $w\in \T\setminus\{\pm 1\}$ is a root, so is $1/w$, in both cases with the same multiplicities. Requiring that these symmetries be respected in \eqref{eq:PQ2} shows that the polynomial $Q(z)/\bar{Q}$ obtained from \eqref{eq:QFR} has real coefficients. Choosing $\bar{Q}\in \R$ gives the result.

\subsection{Proof of Corollary~\ref{cor:main}}\label{subsec:corproof}

Performing the change of variables $z=\ee^{\ii\theta}$, $z'=\ee^{\ii\theta'}$ in \eqref{eq:QT} gives
\begin{equation}\label{eq:Qexp}
Q(\ee^{\ii\theta})=Q_0(\ee^{\ii\theta})\exp\Bigg(\frac{1}{4\pi\ii}\,\pvint_{-\pi}^{\pi} \cot\bigg(\frac{\theta'-\theta}{2}\bigg)	\log \bigg(\frac{1-\lvert P(\ee^{\ii\theta'})\rvert^2}{\lvert Q_0(\ee^{\ii\theta'})\rvert^2}\bigg)\,\mathrm{d}\theta'+\frac12 	\log \bigg(\frac{1-\lvert P(\ee^{\ii\theta})\rvert^2}{\lvert Q_0(\ee^{\ii\theta})\rvert^2}\bigg)\Bigg).
\end{equation}
The first term in the exponent is identified as a periodic Hilbert transform \cite{king2009},
\begin{equation}\label{eq:H}
H[f(\theta)]\coloneqq \frac{1}{2\pi}\,\pvint_{-\pi}^{\pi} \cot\bigg(\frac{\theta'-\theta}2\bigg)f(\theta')\,\mathrm{d}\theta'.
\end{equation} 
We recall that the periodic Hilbert transform \eqref{eq:H} has the complex exponentials as eigenfunctions, 
\begin{equation}\label{eq:Hexp}
H[\ee^{\ii n\theta}]=\begin{cases}
\ii\ee^{\ii n\theta} & n\in \Z_{\geq 1} \\
0 & n=0 \\
-\ii \ee^{\ii n\theta} & n\in \Z_{\geq 1}.	
\end{cases}
\end{equation}
Writing $\Pi=\frac12(1-\ii H)$, we see from \eqref{eq:Hexp} that \eqref{eq:Piexp} holds. Expressing \eqref{eq:Qexp} in terms of $\Pi$ gives the result.

\section{Numerical methods}
\label{sec:numerical}

We develop a numerical method for solving Problem~\ref{prob:CP} based on Corollary~\ref{cor:main} in the case $Q_0(z)=1$. Our starting point is the Laurent series
\begin{equation}\label{eq:S}
S(z)\coloneqq \sum_{n\in \Z} a_n z^n,
\end{equation}
where
\begin{equation}\label{eq:an}
a_n\coloneqq \frac{1}{2\pi\ii}\int_{\T} \log\big({1-\lvert P(z)\rvert^2}\big)\frac{\mathrm{d}z}{z^{n+1}} \quad (n\in \Z). 
\end{equation}
Observe that $S(\ee^{\ii\theta})$ is the Fourier series of $\log\big({1-\lvert P(\ee^{\ii\theta})\rvert^2}\big)$. Thus, from \eqref{eq:QFourier} and \eqref{eq:Piexp}, we have
\begin{equation}\label{eq:Qexp2}
Q(\ee^{\ii\theta})=\exp\big(\Pi[{S}(\ee^{\ii\theta})]\big)=\exp\Bigg(\frac12 a_0+\sum_{n=1}^{\infty} a_n\ee^{\ii n\theta}\Bigg).	
\end{equation}

To numerically evaluate \eqref{eq:Qexp2}, we make two approximations that allow us to compute $Q$ in the monomial basis by a sequence of FFTs. As the performance of an FFT is optimized when the number of Fourier modes is a power of $2$, we choose the size of this basis to be $N=2^M$ for some $M\in \Z_{\geq 1}$ satisfying $M\geq \lceil \log_2(d+1)\rceil$. The error analysis of Algorithm~\ref{algo:main}, which will result from the approximations and analysis in this section, is performed in Section~\ref{sec:error}. First, we introduce the Laurent polynomial truncation of \eqref{eq:S},
\begin{equation}\label{eq:SN}
S_N(z)\coloneqq \sum_{n=-\frac{N}2+1}^{\frac{N}{2}} a_n z^n \quad (N\in \Z_{\geq d_1+1}).
\end{equation}
Second, we will approximate the coefficients \eqref{eq:an} by discrete Fourier transforms. Consider the primitive $N$th root of unity
\begin{equation}
\omega_{N}\coloneqq \ee^{2\pi\ii/N}, 
\end{equation}
which we use to define the following approximation of the Laurent coefficients \eqref{eq:an},
\begin{equation}\label{eq:antilde}
\tilde{a}_n\coloneqq \frac{1}{N}\sum_{m=-\frac{N}{2}+1}^{\frac{N}{2}} \log\big({1-\lvert P(\omega_{N}^{m})\rvert^2}\big)\omega_{N}^{-nm} \quad (n=-\tfrac{N}{2}+1,\ldots,\tfrac{N}{2}).	
\end{equation}	
It follows that
\begin{equation}\label{eq:SNtilde}
\tilde{S}_N(z)\coloneqq \sum_{n=-\frac{N}{2}+1}^{\frac{N}2} \tilde{a}_n z^n	
\end{equation}
is an approximation of $S_N$ \eqref{eq:SN}. \\

Replacing $S(\ee^{\ii\theta})$ by $\tilde{S}_N(\ee^{\ii\theta})$ in \eqref{eq:Qexp2} gives
\begin{equation}\label{eq:QN}
\tilde{Q}_{1,N}(\ee^{\ii\theta})\coloneqq \exp\big(\Pi\big[\tilde{S}_N(\ee^{\ii\theta})\big]\big)=	\exp\Bigg(\frac12 \tilde{a}_0+\sum_{n=1}^{\frac{N}{2}} \tilde{a}_n\ee^{\ii n\theta}\Bigg)
\end{equation}
as an approximation of $Q$ on $\T$. This is, however, no guarantee that $\tilde{Q}_{1,N}(\ee^{\ii\theta})$ is a trigonometric polynomial or equivalently, extends to a polynomial $\tilde{Q}_{1,N}(z)$ on $\C$. We can instead (i) interpolate \eqref{eq:QN} through the roots of unity $\{\omega_{N}^n\}_{n=0}^{N-1}$ and (ii) discard terms in $z^n$ for $n>d$ to obtain an explicit polynomial of degree $d$ in the monomial basis. \\

Let us write
\begin{equation}\label{eq:Qpoly}
Q(z)=\sum_{n=0}^d q_n z^n,	
\end{equation}
in correspondence with \eqref{eq:P}. We define $q_n=0$ for $n>d$. The evaluation of $Q$ at the roots of unity $\{\omega_{N}^n\}_{n=0}^{N-1}$,
\begin{equation}
Q(\omega_{N}^n)=\sum_{m=0}^{N-1} q_m\omega_{N}^{nm} \quad (n\in [N-1]_0),
\end{equation}
is an inverse discrete Fourier transform of the coefficients $(q_n)_{n=0}^{N-1}$. Thus, the corresponding forward transform allows for the computation of $(q_n)_{n=0}^d$,  
\begin{equation}\label{eq:qn}
q_n= \frac{1}{N}\sum_{m=0}^{N-1} Q(\omega_{N}^m)\omega_{N}^{-nm} \quad (n\in [d]_0).
\end{equation}
We are led to define the following approximations of the monomial coefficients $(q_n)_{n=0}^{d}$,
\begin{equation}\label{eq:qntilde}
\tilde{q}_n\coloneqq \frac{1}{N}\sum_{m=0}^{N-1} \tilde{Q}_{1,N}(\omega_{N}^m)\omega_{N}^{-nm} \quad (n\in [d]_0)
\end{equation}
and the following manifestly polynomial approximation to $Q$,
\begin{equation}\label{eq:Q2}
\tilde{Q}_{2,N}(z)\coloneqq \sum_{n=0}^d \tilde{q}_n z^n. 
\end{equation}

\subsection{Algorithm}
\label{sec:algorithm}

We combine the observations obtained in this section into algorithms to compute $\tilde{Q}_{2,N}$ \eqref{eq:Q2}, an approximate canonical complementary polynomial to $P$. In order to avoid the use of root-finding to determine $Q_0$, we will consider situations where $Q_0(z)=1$; see Algorithm~\ref{algo:main}. However, this is not a restriction; by downscaling the input polynomial, $Q_0(z)=1$ can always be achieved. Accordingly, the generalized Algorithm~\ref{algo:downscaling} applies to any target polynomial satisfying $\lVert P(z)\rVert_{\infty,\T}\leq 1$. \\

For many practical applications of QSP-type algorithms in quantum computation, the parameter $\delta$ in \eqref{eq:deltabound}
can be controlled \textit{a priori} in the construction of a polynomial $P$ approximating a target function.
Then, \eqref{eq:deltabound} ensures that $1-|P(z)|^2$ has no roots on $\T$ and hence $Q_0(z)=1$. Algorithm~\ref{algo:main} applies directly in this scenario.

\begin{algo}[Construction of a canonical complementary polynomial for known $\delta$]\label{algo:main} 
\hfill\break
\textbf{Input:} 
\begin{itemize}
    \item The monomial coefficients $(p_n)_{n=0}^{d}$ of $P\in \C[z]$, $\deg P=d$, satisfying \eqref{eq:deltabound} for known $\delta\in (0,1)$.
    \item An integer $N\in \mathbb{Z}_{\geq d}$ determining the dimension of the FFTs and thus controlling the accuracy $\varepsilon$ of the output.
\end{itemize}
\textbf{Output:} 
\begin{itemize}
\item The monomial coefficients $(\tilde{q}_n)_{n=0}^d$ of an approximate canonical complementary polynomial $\tilde{Q}_{2,N}\in\C[z]$, $\deg \tilde{Q}_{2,N}=d$, approximating $Q$ from Theorem~\ref{thm:main} to accuracy $\varepsilon$ in each monomial coefficient; see \eqref{eq:Q2}.
\end{itemize}
\textbf{Complexity:}
\begin{itemize}
    \item Runtime: $O(N\log N)$.
    \item Sufficient $N$ for accuracy $\varepsilon$: $N=O\!\left(\frac{d}{\delta}\log\frac{d}{\delta\varepsilon}\right)$; see Theorem~\ref{thm:error}.
\end{itemize}
\textbf{Algorithm:}
\begin{enumerate}
\item Compute $(P(\omega_N^n))_{n=0}^{N-1}$, the input polynomial $P(z)=\sum_{n=0}^dp_nz^n$ evaluated at all $N$th roots of unity, with the inverse FFT of $(p_n)_{n=0}^{N-1}$ padded with zeros, i.e., $p_n=0$ for $n>d$.
\item Compute $(\tilde{a}_n)_{n=-\frac{N}{2}+1}^{\frac{N}2}$ by applying the FFT to $(\log\left(1-\left|P(\omega_N^n)\right)|^2\right))_{n=0}^{N-1}$ obtained from the previous step, see \eqref{eq:antilde}.
\item Apply the Fourier multiplier $\Pi$ to the truncated Fourier series \eqref{eq:SNtilde}.
\item Evaluate the exponential $(\tilde Q_{1,N}(\omega_N^n))_{n=0}^{N-1}$ in \eqref{eq:QN} at $N$th roots of unity by taking the exponential of an inverse FFT of the previous step's result.
\item Compute $(\tilde q_n)_{n=0}^{N-1}$ in \eqref{eq:qntilde} by applying the FFT to $(\tilde Q_{1,N}(\omega_N^n))_{n=0}^{N-1}$.
\item Truncate the coefficients of the previous step to $(\tilde q_n)_{n=0}^d$ and output them as the coefficients of the approximation $\tilde Q_{2,N}(z)$ to the complementary polynomial, see \eqref{eq:Q2}.
\end{enumerate}
\textbf{Reference implementation:}
\begin{itemize}
    \item See Figure~\ref{fig:code} for Python code and Section~\ref{sec:numerical} for numerical results.
\end{itemize}
\end{algo}

The algorithm relies on FFTs to map between the coefficients of a polynomial and its values at roots of unity and to apply the Fourier multiplier $\Pi$ in Fourier space \eqref{eq:QFourier}--\eqref{eq:Pi description}. The FFTs of dimension $N$ correspond to an overall runtime of $O(N\log N)$. \\

In Theorem~\ref{thm:error}, stated in Section~\ref{sec:error}, we prove that Algorithm~\ref{algo:main} computes a canonical complementary polynomial to accuracy $\varepsilon$ in the monomial coefficients. We emphasize that the canonical complementary polynomial is the unique solution of Problem~\ref{prob:CP}, up to a multiplicative phase, with no roots in $\D$, as in Theorem~\ref{thm:main}. Theorem~\ref{thm:error} moreover shows that the algorithm is efficient in degree and error, with a sufficient $N = O\big(d\log \frac{d}{\varepsilon}\big)$ for fixed $\delta$. \\

For small $\delta$, the scaling $N\sim\frac1{\delta}$ from Theorem~\ref{thm:error} suggests a long runtime; in the extreme case $\delta=0$, the proof for Algorithm~\ref{algo:main} fails because $Q_0(z)\neq1$. 
For those cases, we present Algorithm~\ref{algo:downscaling}, in which the initial polynomial is downscaled as $P(z)\to(1-\tfrac{\varepsilon}4)P(z)$ to achieve an effective $\delta = \frac{\varepsilon}4$. 
In Theorem~\ref{thm:error downscaling}, stated in Section~\ref{sec:error}, we prove that the polynomial generated by Algorithm~\ref{algo:downscaling} with $N=O\big(\frac{d}{\varepsilon}\log\frac{d}{\varepsilon})$ satisfies the complementarity condition \eqref{eq:PQ} to accuracy $\varepsilon$, i.e., $
\big\lVert \lvert P(z)\rvert^2+\lvert Q(z)\rvert^2\big\rVert_{\infty,\T}< \varepsilon$. Closeness in the complementarity condition \eqref{eq:PQ} is a weaker statement than closeness to an exact canonical complementary polynomial, as Theorem~\ref{thm:error} promises for Algorithm~\ref{algo:main}. Yet, it enables us to rigorously and efficiently extend our numerical method to all polynomials $P$ with $\lVert P(z)\rVert_{\infty,\T}\leq 1$. 

\begin{algo}[Construction of a complementary polynomial for zero, unknown, or small $\delta$]
\label{algo:downscaling}
\hfill\break
\textbf{Input:} 
\begin{itemize}
    \item The monomial coefficients of $P\in \C[z]$, $\deg P=d$, with $\lVert P(z)\rVert_{\infty,\T}\leq 1$.
    \item An integer $N\in \mathbb{Z}_{\geq d}$ determining the dimension of the FFTs and thus controlling the accuracy $\varepsilon$ of the output.
\end{itemize}
\textbf{Output:} 
\begin{itemize}
\item The monomial coefficients $(\tilde{q}_n)_{n=0}^d$ of a canonical complementary polynomial $\tilde{Q}_{2,N}\in \C[z]$, $\deg \tilde{Q}_{2,N}=d$, to accuracy $\varepsilon$ in the complementarity condition \eqref{eq:PQ}, i.e., $
\big\lVert \lvert P(z)\rvert^2+\lvert Q(z)\rvert^2\big\rVert_{\infty,\T}< \varepsilon$.
\end{itemize}
\textbf{Complexity:}
\begin{itemize}
    \item Runtime: $O(N\log N)$.
    \item Sufficient $N$ for accuracy $\varepsilon$: $N=O\!\left(\frac{d}{\varepsilon}\log\frac{d}{\varepsilon}\right)$; see Theorem~\ref{thm:error downscaling}.
\end{itemize}
\textbf{Algorithm:}
\begin{enumerate}
\item Compute $\big((1-\tfrac{\varepsilon}4)p_n\big)_{n=0}^d$ to scale down the input polynomial $P(z)=\sum_{n=0}^d p_nz^n$.
\item Return the result of Algorithm~\ref{algo:main} with input the downscaled polynomial from the previous step, for which $\delta=\frac{\varepsilon}4$, and $N$ chosen to yield an accuracy of $\frac{\varepsilon}{5(d+1)}$.
\end{enumerate}
\end{algo}

\begin{rem}
 Empirically, we find that Algorithm~\ref{algo:downscaling} may not be needed. Even without the initial downscaling, our numerical results in Section~\ref{sec:numericalresults} suggest that Algorithm~\ref{algo:main} alone is efficacious even when $\delta=0$, with a sufficient $N= O\big(\frac{d}{\sqrt[4]\varepsilon}\big)$. Here, and in our error analysis in the next section, we include Algorithm~\ref{algo:downscaling} to preserve full mathematical rigor.   
\end{rem}

\section{Error analysis of algorithms}
\label{sec:error}

We perform error analysis on Algorithm~\ref{algo:main} and Algorithm~\ref{algo:downscaling}, developed in the previous section. Our main results, Theorem~\ref{thm:error}, Corollary~\ref{cor:error}, and Theorem~\ref{thm:error downscaling}, are stated below. The corresponding proofs are given in Sections~\ref{subsec:errorproof}, \ref{subsec:errorcorproof}, and \ref{subsec:downscaling proof}, respectively.

\subsection{Error metrics}
\label{subsec:metrics}

We introduce two error metrics that we will later use to analyze the algorithms. In the definitions of these error metrics, we view $P,Q\in \C[z]$ as \textit{a priori} unrelated; when $Q$ is an exact complementary polynomial to a given $P$, the error metrics will evaluate to zero. \\

The first error metric we consider is motivated by the complementarity condition \eqref{eq:PQ}, 
\begin{equation}\label{eq:Phi}
\Phi(P,Q)\coloneqq \big\lVert \lvert P(z)\rvert^2+\lvert Q(z)\rvert^2-1\big\rVert_{\infty,\T}.    
\end{equation}
The second error metric we consider was introduced in \cite{motlagh2024} as a loss function for optimization of $Q$ and is defined in terms of the monomial coefficients of the polynomials $P$ and $Q$,
\begin{equation}\label{eq:loss}
\tilde{\Phi}(P,Q)\coloneqq \Bigg(\sum_{n=-d}^d \Bigg\lvert \sum_{m=0}^d\big( p_{n+m} p_{m}^*+{q}_{n+m}{q}_{m}^*\big) -{\delta_{n,0}}  \Bigg\rvert^2\Bigg)^{\frac12}.
\end{equation}
The error metrics \eqref{eq:Phi} and \eqref{eq:loss} are compatible in the following sense. 

\begin{proposition}
\label{prop:norm equivalent}
Let $P,Q\in \C[z]$ with $\deg P=\deg Q=d$. Then, the complementarity condition $\Phi(P,Q)$ and the loss function $
\tilde{\Phi}(P,Q)$ are equivalent in the sense that they satisfy the inequalities   \begin{equation}\label{eq:norminequalities}
\frac{1}{\sqrt{2d+1}}\Phi(P,Q)\leq  \tilde{\Phi}(P,Q)\leq \sqrt{2d+1}\,\Phi(P,Q).
\end{equation} 
\end{proposition}

\begin{proof}

Putting \eqref{eq:P} and \eqref{eq:Qpoly} into \eqref{eq:Phi}, we write
\begin{equation}
\Phi(P,Q)= \Bigg\lVert \sum_{n,m=1}^d \big(p_np_m^*+q_nq_m^* \big) z^{n-m}-1\Bigg\rVert_{\infty,\T}.  
\end{equation}
By changing the summation variables and using the triangle and $\ell^1$-$\ell^2$ norm inequalities, we obtain
\begin{align}\label{eq:Phicalc}
\Phi(P,Q)=&\; \Bigg\lVert \sum_{n=-d}^d\Bigg(\sum_{m=0}^d (p_{n+m}p_m^*+q_{n+m}q_m^*\big)-\delta_{n,0}\Bigg)z^n  \Bigg\rVert \leq  \sum_{n=-d}^d\Bigg\lvert \sum_{m=0}^d  (p_{n+m}p_m^*+q_{n+m}q_m^*\big)-\delta_{n,0} \Bigg\rvert \nonumber \\
 \leq &\; \sqrt{2d+1} \Bigg(\sum_{n=-d}^d\Bigg\lvert \sum_{m=0}^d  (p_{n+m}p_m^*+q_{n+m}q_m^*\big)-\delta_{n,0} \Bigg\rvert^2 \Bigg)^{\frac12},
\end{align}
which, recalling \eqref{eq:loss}, is the first inequality in \eqref{eq:norminequalities}.  \\

To prove the second inequality in \eqref{eq:norminequalities}, we use  Cauchy integral formula to write
\begin{equation}
\sum_{m=0}^d \big(p_{n+m}p_m^*+q_{n+m}q_m^*\big)-\delta_{n,0}=\frac{1}{2\pi\ii} \int_{\T} \Bigg( \sum_{l=-d}^d \sum_{m=0}^d  (p_{l+m}p_m^*+q_{l+m}q_m^*-\delta_{l,0}) z^n\Bigg) \frac{\mathrm{d}z}{z^{n+1}} \quad (n=-d,\ldots,d) 
\end{equation}
and hence,
\begin{equation}\label{eq:coefftoPhi}
\Bigg\lvert \sum_{m=0}^d \big(p_{n+m}p_m^*+q_{n+m}q_m^*\big)-\delta_{n,0}\Bigg\rvert \leq \Phi(P,Q) \quad (n=-d,\ldots,d) , 
\end{equation}
where we have used \eqref{eq:Phicalc}. Putting \eqref{eq:coefftoPhi} into \eqref{eq:loss} gives the result.
\end{proof}

\subsection{Results of error analysis}

We establish rigorous error bounds on Algorithms~\ref{algo:main} and \ref{algo:downscaling}, using the error metrics introduced in the previous subsection. 

\begin{theorem}[Error bounds for Algorithm~\ref{algo:main}]
\label{thm:error}
Suppose $P\in \C[z]$ satisfying \eqref{eq:deltabound} for some $\delta\in (0,1)$ and $\varepsilon\in (0,1)$ are given. Choose $N\in \Z_{\geq 1}$ such that
\begin{equation}\label{eq:Nbound}
N \geq N_0(\varepsilon,\delta,d)\coloneqq \bigg\lceil  \frac{2}{\log r_{\delta}}\log\bigg(8 \frac{\log(\frac{1}{\delta})}{r_{\delta}-1}\frac{1}{\varepsilon} \bigg)\bigg\rceil,
\end{equation}
where
\begin{equation}\label{eq:repsdel}
r_{\delta}\coloneqq \bigg(\frac{1}{1-\delta}\bigg)^{\frac1d}.   
\end{equation}
Then, the output of Algorithm~\ref{algo:main} satisfies 
\begin{equation}\label{eq:QQ2result}
\lvert q_n-\tilde{q}_n\rvert < \varepsilon \quad (n\in [d]_0).
\end{equation}
In particular, \eqref{eq:Nbound} has the joint asymptotic complexity
\begin{align}
\label{eq:algo1 complexity}
N_0(\varepsilon,\delta,d) = O\bigg( \frac{d}{\delta}\log\frac{d}{\delta\varepsilon}\bigg).
\end{align}
\end{theorem}

We can use Theorem~\ref{thm:error} to bound the error metrics introduced in Section~\ref{subsec:metrics}

\begin{corollary}
\label{cor:error}
Suppose that \eqref{eq:QQ2result} holds for some $\varepsilon\in (0,1)$. Then, the error metrics \eqref{eq:Phi} and \eqref{eq:loss} satisfy the inequalities
\begin{equation}\label{eq:Phibound}
\Phi(P,\tilde{Q}_{2,N}) < (d+1)(d+3)\varepsilon 
\end{equation}
and 
\begin{equation}\label{eq:objbound}
\tilde\Phi(P,\tilde{Q}_{2,N}) < 3(d+1)(2d+1)\varepsilon.
\end{equation}
\end{corollary}

Letting $\tilde{\varepsilon}=\tilde\Phi(P,\tilde{Q}_{2,N})$, we see from Theorem~\ref{thm:error} that for fixed $\delta$, $N=O\big(d \log\frac{1}{\tilde{\varepsilon}}\big)$ is required to achieve \eqref{eq:objbound}. Numerical results verifying this assertion are presented in Section~\ref{sec:numericalresults}.

\begin{theorem}[Error bounds for Algorithm~\ref{algo:downscaling}]
\label{thm:error downscaling}
Suppose $P\in \C[z]$ satisfying $\lVert P(z)\rVert_{\T,\infty}\le1$ and $\varepsilon\in(0,1)$ are given. Choose $N\in \Z_{\geq 1}$ such that 
\begin{equation}
N\geq N_0\bigg(\frac{\varepsilon}{4},\frac{\varepsilon}{5(d+1)},d\bigg)    
\end{equation}
with $N_0$ defined in \eqref{eq:Nbound}. Then, the result $\tilde Q_{2,N}\in\C[z]$ of Algorithm~\ref{algo:downscaling} satisfies the bound
\begin{equation}
    \Phi(P,\tilde Q_{2,N}) < \varepsilon
\end{equation}
on the complementarity condition \eqref{eq:Phi}.
In particular, we have the joint asymptotic complexity
\begin{equation}
    \label{eq:algo2 complexity}
   N_0\bigg(\frac{\varepsilon}{4},\frac{\varepsilon}{5(d+1)},d\bigg)=O\bigg(\frac{d}{\varepsilon}\log\frac{d}{\varepsilon}\bigg).
\end{equation}
\end{theorem}

\subsection{Proof of Theorem~\ref{thm:error}}
\label{subsec:errorproof}

Let 
\begin{equation}R\coloneqq \min_{j\in [d_1]} \lvert w_j\rvert 
\end{equation}
and define the function
\begin{equation}\label{eq:M}
M(r)\coloneqq \max_{\rho=\frac1r,r}\max_{z\in \T}\big\lvert \log \big({1- P(\rho z)P^*(1/\rho z)}\big)\big\rvert \quad (r\in (1,R)).
\end{equation}
Our analysis is based on the following lemma.

\begin{lemma}\label{lem:an}
For $r\in (1,R)$, the Fourier coefficients $(a_n)_{n\in \Z}$ from \eqref{eq:an} satisfy 
\begin{equation}\label{eq:anbound}
\lvert a_n\rvert \leq M(r)r^{-\lvert n\rvert} \quad 	(n\in \Z).
\end{equation}
\end{lemma}

\begin{proof}
For any $r\in (1,R)$, the function $\log \big({1-\lvert P(z)\rvert^2}\big)$ may be analytically continued to the closure of the annulus
\begin{equation}\label{eq:annulus}
A(r)\coloneqq \{z\in \C: \tfrac1{r}< \lvert z\rvert <  r\}.
\end{equation}
Suppose $n\in \Z_{\geq 0}$. Then, using Cauchy's theorem to deform the contour in \eqref{eq:an}, we find
\begin{align}\label{eq:anpositive}
\lvert a_n \rvert = &\; \frac{1}{2\pi}\Bigg\lvert \int_{\lvert z\rvert=r}\log \big({1-P(z)P^*(1/z)}\big)\frac{\mathrm{d}z}{z^{n+1}}\Bigg\rvert = r^{-n}\frac{1}{2\pi}\Bigg\lvert \int_{\T}\log\big({1-P(rz)P^*(1/rz)}\big)\frac{\mathrm{d}z}{z^{n+1}}\Bigg\rvert \nonumber \\
\leq &\; r^{-n}\frac{1}{2\pi}\int_{\T} \bigg\lvert \log\big({1-P(rz)P^*(1/rz)}\big)\frac{1}{z^{n+1}}\Bigg\rvert \,\mathrm{d}z = r^{-n}\frac{1}{2\pi}\int_{\T} \big\lvert \log\big({1-P(rz)P^*(1/rz)}\big)\big\rvert \,\mathrm{d}z \leq  L(r) r^{-n},
\end{align}
where
\begin{equation}
L(r)\coloneqq \sup_{\frac{1}{r}<\rho< r} \frac{1}{2\pi}\int_{\T} \big\lvert \log \big({1- P(\rho z)P^*(1/\rho z)}\big)\big\rvert \,\mathrm{d}z \quad (r \in (1,R)),
\end{equation}
for each $n\in \Z_{\geq 0}$. A similar argument for $n\in \Z_{\leq 0}$ shows that
\begin{equation}\label{eq:annegative}
\lvert a_n\rvert \leq L(r)r^n \quad (n\in \Z_{\leq 0}). 
\end{equation}

By the maximum modulus principle, we have
\begin{equation}\label{eq:LM}
L(r)\leq M(r) \quad (r\in (1,R)). 	
\end{equation}
The result \eqref{eq:anbound} follows by combining \eqref{eq:anpositive} and \eqref{eq:annegative} with \eqref{eq:LM}.
\end{proof}

Using Lemma~\ref{lem:an}, we readily obtain a bound on the truncation error, 
\begin{equation}\label{eq:truncation}
\big\lvert \Pi[S(\ee^{\ii\theta})]-\Pi[S_N(\ee^{\ii\theta})]\big\rvert=\Bigg\lvert \sum_{n=\frac{N}{2}+1}^{\infty}a_n \ee^{\ii n\theta}\Bigg\rvert \leq \sum_{n=\frac{N}{2}+1}^{\infty}\lvert a_n\rvert \leq M(r)\sum_{n=\frac{N}{2}+1}^{\infty} r^{-n}=\frac{M(r)}{r^{N/2}(r-1)}	\quad (\theta\in (-\pi,\pi]).
\end{equation}

To obtain a corresponding bound for the difference between $\Pi[S_N(\ee^{\ii\theta})]$ and $\Pi[\tilde{S}_N(\ee^{\ii\theta})]$, we recall the discrete Poisson summation formula \cite[Chapter~6]{briggs1995}
\begin{equation}\label{eq:Poisson}
\tilde{a}_n=a_n+\sum_{m\in \Z\setminus\{0\}} a_{n+Nm}.
\end{equation}
Together, \eqref{eq:anbound} and \eqref{eq:Poisson} give
\begin{equation}
\lvert \tilde{a}_n-a_n\rvert = \Bigg\lvert \sum_{m\in \Z\setminus\{0\}} a_{n+Nm}\Bigg\rvert \leq \sum_{m\in \Z\setminus\{0\}} \lvert a_{n+Nm} \rvert \leq 2 M(r)r^{-n}\sum_{m=1}^{\infty} r^{-Nm}=\frac{2M(r)}{r^n(r^N-1)} \quad (\lvert n\rvert \leq N). 
\end{equation}
It then follows that
\begin{align}\label{eq:PiSNPiSNtilde}
\big\lvert \Pi[S_N(\ee^{\ii\theta})]-\Pi[\tilde{S}_N(\ee^{\ii\theta})]\big\rvert= &\; \Bigg\lvert \frac12(a_0-\tilde{a}_0)+\sum_{n=1}^{N}(a_n-\tilde{a}_n) \ee^{\ii n\theta}\Bigg\rvert \leq \frac12\lvert a_0-\tilde{a}_0\rvert +\Bigg\lvert \sum_{n=1}^{N}(a_n-\tilde{a}_n) z^n\Bigg\rvert \nonumber \\
\leq &\; \frac12\lvert a_0-\tilde{a}_0\rvert +\sum_{n=1}^{N}\lvert a_n-\tilde{a}_n\rvert \leq  \frac{M(r)}{r^N-1}\Bigg(1+2\sum_{n=1}^{N} r^{-n}\Bigg)\nonumber  \\
=&\; \frac{M(r)}{r^N-1}\bigg(1+2\frac{r^N-1}{r^N(r-1)} \bigg)   = M(r)\frac{r^{N+1}+r^N-2}{r^N(r^N-1)(r-1)} < M(r) \frac{r+2}{r^N(r-1)} \quad (\theta\in (-\pi,\pi]),
\end{align}
where we have used that $(r^{N+1}-1)/(r^N-1)<r+1$ in the final step. Hence, \eqref{eq:truncation} and \eqref{eq:PiSNPiSNtilde} and the triangle inequality imply
\begin{equation}\label{eq:PiSPiStilde}
\big\lvert \Pi[S(\ee^{\ii\theta})]-\Pi[\tilde{S}_N(\ee^{\ii\theta})]\big\rvert \le M(r)\frac{r^{N/2} + r +2}{r^N(r-1)} < \frac{4M(r)}{r^{N/2}(r-1)} \quad (\theta\in (-\pi,\pi]). 
\end{equation}

We can now compute
\begin{align}\label{eq:QQtilde}
\big\lvert Q(\ee^{\ii\theta})-\tilde{Q}_{1,N}(\ee^{\ii\theta}) \big\rvert  = &\; \big\lvert \exp\big(\Pi[S(\ee^{\ii\theta})]\big)-\exp\big(\Pi[\tilde{S}_N(\ee^{\ii\theta})]\big) \big\rvert  	\nonumber \\
=&\; \big\lvert \exp\big(\Pi[S(\ee^{\ii\theta})]\big)\big(1-\exp\big(\Pi[\tilde{S}_N(\ee^{\ii\theta})]-\Pi[{S}(\ee^{\ii\theta})]\big)\big\rvert \nonumber \\
< &\;  \exp \bigg(  \frac{4M(r)}{r^{N/2}(r-1)}	\bigg)-1 \quad (\theta\in (-\pi,\pi]),
\end{align}
where we have used $\lvert Q(\ee^{\ii\theta})\rvert = \big\lvert \exp\big(\Pi[S(\ee^{\ii\theta})]\big)\big\rvert < 1$ and \eqref{eq:PiSPiStilde} in the final step.\\

Next, we use \eqref{eq:qn}, \eqref{eq:qntilde}, and \eqref{eq:QQtilde} to write
\begin{equation}\label{eq:qnqtilden}
\lvert q_n-\tilde{q}_n\rvert= \frac{1}{N}\Bigg\lvert \sum_{m=0}^{N-1} \big(Q(\omega_{N}^{m})-\tilde{Q}_{1,N}(\omega_{N}^m)\big) \omega_{N}^{-nm}\Bigg\rvert < \exp \bigg( \frac{4M(r)}{r^{N/2}(r-1)}	\bigg)-1 \quad (n\in [d]_0).	
\end{equation}

We are guaranteed that the argument of the exponential in \eqref{eq:qnqtilden} is upper-bounded by unity provided that 
\begin{equation}\label{eq:Nbound1}
N \geq \frac{2}{\log r}\log\bigg(\frac{4M(r)}{r-1}\bigg).	
\end{equation}
Suppose that \eqref{eq:Nbound1} holds. Then, using $\ee^x-1<2x$ for $x\in (0,1)$ we have 
\begin{equation}\label{eq:expbound}
\exp \bigg( \frac{4M(r)}{r^{N/2}(r-1)}	\bigg)-1 < \frac{8M(r)}{r^{N/2}(r-1)}.
\end{equation}
Putting \eqref{eq:expbound} in \eqref{eq:qnqtilden} yields
\begin{equation}\label{eq:QQ2bound}
\lvert q_n-\tilde{q}_n\rvert 	< 8M(r)\frac{1}{r^{N/2}(r-1)} \quad (n\in [d]_0)
\end{equation}
and we find that \eqref{eq:QQ2bound} is upper-bounded by $\varepsilon$ provided that
\begin{equation}\label{eq:QQ2N}
N\geq  \frac{2}{\log r}\log\bigg(\frac{8M(r)}{r-1}\frac{1}{\varepsilon}\bigg),
\end{equation}
which implies \eqref{eq:Nbound1}, holds. We can make \eqref{eq:QQ2N} more precise by specifying an $r\in (1,R)$ and bounding $M(r)$. We choose $r=r_{\delta}$, defined in \eqref{eq:repsdel}. We have $r_{\delta}>1$ by the assumption that $\delta\in (0,1)$. The following lemma shows that $r_{\delta}<R$ by bounding $\lvert P(z)P^*(1/z)\rvert $ in the annulus $A(r_{\delta})$ \eqref{eq:annulus}.

\begin{lemma}
The following inequality holds,
\begin{equation}
\label{eq:PPbound}
\lvert P(z)P^*(1/z)\rvert \leq  1-\delta\quad(z\in A(r_\delta)).
\end{equation}
\end{lemma}

\begin{proof}
We define the reciprocal polynomial to $P$ by
\begin{equation}
{P}^{\mathrm{R}}(z)\coloneqq z^d P^*(1/z);    
\end{equation}
${P}^{\mathrm{R}}$ is a polynomial and hence entire. Note that $\lVert {P}^{\mathrm{R}}(z)\rVert_{\infty, \T}=\lVert P(z)\rVert_{\infty,\T}\leq 1-\delta$. It follows, by the maximum modulus principle, that
\begin{equation}\label{eq:PPhatbounds}
 \lvert P(z) \rvert,  \lvert {P}^{\mathrm{R}}(z) \rvert \leq  1-\delta \quad (z\in \D). 
\end{equation}
Using the conformal map $z\mapsto 1/z$, we deduce from \eqref{eq:PPhatbounds} the corresponding bounds
\begin{equation}\label{eq:PPhatbounds2}
\lvert P(1/z)\rvert , \lvert {P}^{\mathrm{R}}(1/z)\rvert \leq  1-\delta \quad (z\in \C\setminus\overline{\D}).
\end{equation}

Because $P(z)=z^d ({P}^{\mathrm{R}})^*(1/z)$, we have 
\begin{equation}\label{eq:PCDbound}
\lvert P(z)\rvert = \lvert z \rvert^d \lvert{P}^{\mathrm{R}}(1/z)\rvert   \le   (1-\delta)\lvert z\rvert^d \quad (z\in \C\setminus\overline{\D}),
\end{equation}
where we have used \eqref{eq:PPhatbounds2}. Again using $z\mapsto 1/z$, we have
\begin{equation}\label{eq:PCbound}
\lvert P(1/z)\rvert \leq    (1-\delta)\lvert z\rvert^{-d} \quad (z\in \D).
\end{equation}
Combining \eqref{eq:PPhatbounds}--\eqref{eq:PPhatbounds2} and \eqref{eq:PCDbound}--\eqref{eq:PCbound} yields
\begin{equation}
\lvert P(z)P^*(1/z)\rvert \le  (1-\delta)^2\begin{cases} \lvert z\rvert^{-d}   & z\in \D \\
\lvert z\rvert^{d} &  z\in \C\setminus\overline{\D}.
\end{cases}
\end{equation}
Within the annulus $A(r_\delta)$, this implies
\begin{equation}
    |P(z)P^*(1/z)|\le(1-\delta)^2\frac{1}{1-\delta} = 1-\delta,
\end{equation}
as desired. 
\end{proof}

The next lemma provides an estimate for $M(r_{\delta})$ \eqref{eq:M}.

\begin{lemma}
The following bound holds,
\begin{equation}\label{eq:Mrbound}
M(r_{\delta}) \leq   \log\big(\tfrac{1}{\delta}\big).
\end{equation}
\end{lemma}

\begin{proof}
We write
\begin{equation}
M(r_{\delta})= \max_{z\in \partial A(r_{\delta})}  \big\lvert \log \big(1-P(z)P^*(1/z)\big)\big\rvert. 
\end{equation}
Due the estimate \eqref{eq:PPbound}, which guarantees $\lvert P(z)P^*(1/z)\rvert\leq 1-\delta <1$ on $\overline{A(r_{\delta})}$, we may use the Maclaurin series for $\log(1-z)$ and \eqref{eq:PPbound} to write  
\begin{align}
\big\lvert \log \big(1-P(z)P^*(1/z)\big)\big\rvert=&\; \Bigg\lvert \sum_{n=1}^{\infty} \frac{(P(z)P^*(1/z))^n}{n}  \Bigg\rvert \leq \sum_{n=1}^{\infty} \frac{\lvert P(z)P^*(1/z)\rvert^n}{n}\nonumber \\
\le&\; \sum_{n=1}^{\infty} \frac{(1-\delta)^n }{n} =\log \big(\tfrac{1}{\delta}\big) \quad (z\in \overline{A(r_{\delta})});
\end{align}
the result \eqref{eq:Mrbound} follows. 
\end{proof}

Putting \eqref{eq:repsdel} and \eqref{eq:Mrbound} into \eqref{eq:QQ2N} gives the result \eqref{eq:Nbound}. \\
 
Next, we analyze the asymptotic behavior of our bound \eqref{eq:Nbound} on a sufficient $N$. To this end, we introduce the parameter
\begin{equation}
c_\delta \coloneqq \frac{1}{\log\frac{1}{1-\delta}}
\end{equation}
and compute
\begin{align}
    \label{eq:Nexpanded}
    N_0(\varepsilon,\delta,d) = \bigg\lceil 2dc_\delta \log\frac{1}{\varepsilon}+ 2dc_\delta \log\left(8\log\frac{1}{\delta}\right) + 2dc_\delta\log\left(\frac{1}{r_\delta-1}\right)\bigg\rceil 
\end{align}
Immediately, the asymptotic
\begin{equation}
\label{eq:scalingepsproof}
N_0(\varepsilon,\delta,d)  \sim 2dc_\delta \log\frac{1}{\varepsilon} = O\!\left(\log\frac{1}{\varepsilon}\right),\quad (\varepsilon\downarrow 0, \delta,d\ \text{fixed})
\end{equation}
follows; this regime describes increasing the accuracy of the algorithm for a fixed polynomial. \\

Next, note that $c_\delta d\to+\infty$ as $\delta\downarrow 0$ or $d\to\infty$, such that we have the asymptotic equation
\begin{equation}
    \frac{1}{r_\delta - 1} = \frac{1}{\exp(\frac{1}{c_\delta d})-1} \sim c_\delta d \quad (\delta\downarrow 0\  \text{or}\ d \to\infty).
\end{equation}
We insert this into \eqref{eq:Nexpanded}, but keep $\varepsilon$ because we will take a limit $\varepsilon\downarrow 0$ later:
\begin{align}
    \label{eq:Nexpanded2}
    N_0(\varepsilon,\delta,d)  &\sim 2dc_\delta \log\frac{1}{\varepsilon}+ 2dc_\delta \log\left(8\log\frac{1}{\delta}\right) + 2dc_\delta\log\left(c_\delta d\right) \quad (\delta\downarrow 0\  \text{or}\ d \to\infty).
\end{align}
The middle term is subdominant. In the case $d\to\infty$, the middle logarithm is a constant, and in the case $\delta\downarrow 0$ note that $c_\delta\sim1/\delta$.
Dropping the middle term results in
\begin{equation}
N_0(\varepsilon,\delta,d)  \sim 2dc_\delta \log\left(\frac{c_\delta d}{\varepsilon}\right) \quad (\delta\downarrow 0\  \text{or}\ d \to\infty). 
\end{equation}
For fixed $\delta$, we retrieve the joint asymptotic 
\begin{align}
N_0(\varepsilon,\delta,d)  \sim 2dc_\delta \log\frac{d}{\varepsilon} = O\!\left(d\log\frac{d}{\varepsilon}\right)\quad (\varepsilon\downarrow 0, d\to\infty,\delta\ \text{fixed}). 
\end{align}
As $\delta\downarrow 0$, we get
\begin{equation}
    N_0(\varepsilon,\delta,d) \sim 2\frac{d}{\delta}\log\left(\frac{d}{\varepsilon\delta}\right) \quad (\delta\downarrow 0).
\end{equation}
which is valid regardless of whether $\varepsilon\downarrow 0$ or $d\to\infty$.
In particular, \eqref{eq:algo1 complexity} holds provided $\varepsilon\downarrow 0$, $\delta\downarrow 0$, or $d\to\infty$.

\subsection{Proof of Corollary~\ref{cor:error}}
\label{subsec:errorcorproof}

Let $Q$ be the exact complementary polynomial to $P$ obtained from Corollary~\ref{cor:main} in the form \eqref{eq:Qpoly}. \\

We first prove \eqref{eq:Phibound}. From \eqref{eq:QQ2result}, we have
\begin{equation}
\big\lVert Q(z)-\tilde{Q}_{2,N}(z)\big\rVert_{\T,\infty}< (d+1)\varepsilon	
\end{equation}
and hence,
\begin{align}
\label{eq:QQ2calc}
    \big\lVert |Q(z)|^2 - |\tilde{Q}_{2,N}(z)|^2\big \rVert_{\infty,\T}
    \leq &\;  \big\lVert 2Q(z) - Q(z) + \tilde Q_{2,N}(z)\big\rVert_{
   \infty,\T} \big\lVert Q(z) -\tilde{Q}_{2,N}(z)\big\rVert_{\infty,\T} \nonumber \\
       <&\;  (2+(d+1)\epsilon)(d+1)\varepsilon< (d+1)(d+3)\varepsilon.
\end{align}
It follows from \eqref{eq:Phi} and \eqref{eq:QQ2calc} that
\begin{align}
\Phi(P,\tilde{Q}_{2,N})=&\; \big\lVert \lvert P(z)\rvert^2+\lvert \tilde{Q}_{2,N}(z)\rvert^2-1\big\rVert_{\infty,\T}=\big\lVert \big(\lvert P(z)\rvert^2+\lvert Q(z)\rvert^2-1\big)+\big(\lvert \tilde{Q}_{2,N}(z)\rvert^2-\lvert Q(z)\rvert^2\big) \big\rVert_{\infty,\T} \nonumber \\
\leq &\; \Phi(P,Q)+\big\lVert \lvert \tilde{Q}_{2,N}(z)\rvert^2-\lvert Q(z)\rvert^2\big\rVert_{\infty,\T}=\big\lVert \lvert \tilde{Q}_{2,N}(z)\rvert^2-\lvert Q(z)\rvert^2\big\rVert_{\infty,\T}<(d+1)(d+3)\varepsilon, 
\end{align}
which is \eqref{eq:Phibound}.\\

We next prove \eqref{eq:objbound}. Inserting
\begin{align}\label{eq:qq}
\tilde{q}_{n+m}\tilde{q}_m^*=&\; \big({q}_{n+m}+(\tilde{q}_{n+m}-q_{n+m})\big)\big(q_m^*+(\tilde{q}_m^*-q_m^*)\big) \nonumber \\
=&\; q_{n+m}q_m^*+q_{n+m}(\tilde{q}_m^*-q_m^*)+q_m^*(\tilde{q}_{n+m}-q_{n+m})+(\tilde{q}_{n+m}-q_{n+m})(\tilde{q}_m^*-q_m^*);
\end{align}
into the summand of the loss function \eqref{eq:loss} gives
\begin{multline}\label{eq:ppqq}
p_{n+m}p_m^*+\tilde{q}_{n+m}\tilde{q}_m^* -{\delta_{n,0}}=\\\big(p_{n+m}p_m^*+{q}_{n+m}{q}_m^* -{\delta_{n,0}}\big)+q_{n+m}(\tilde{q}_m^*-q_m^*)+ q_m^*(\tilde{q}_{n+m}-q_{n+m})+(\tilde{q}_{n+m}-q_{n+m})(\tilde{q}_m^*-q_m^*).
\end{multline}
By the $\ell^1$--$\ell^2$ norm inequality, we have
\begin{equation}
 \tilde\Phi(P,\tilde{Q}_{2,N})\leq \sum_{n=-d}^d\Bigg\lvert \sum_{m=0}^d \big( p_{n+m}p_m^*+\tilde{q}_{n+m}\tilde{q}_m^*-{\delta_{n,0}} \big)\Bigg\rvert.
\end{equation}
Using \eqref{eq:ppqq}, it follows that
\begin{align}\label{eq:PhiPQ}
\tilde\Phi(P,\tilde{Q}_{2,N})\leq  &\;   \sum_{n=-d}^d\Bigg\lvert \sum_{m=0}^d \big( p_{n+m}p_m^*+{q}_{n+m}{q}_m^*-{\delta_{n,0}} \big)\Bigg\rvert +\sum_{n=-d}^d\Bigg\lvert \sum_{m=0}^d  q_{n+m}(\tilde{q}_m^*-q_m^*)\Bigg\rvert   \nonumber  \\
&\; +\sum_{n=-d}^d\Bigg\lvert \sum_{m=0}^d  q_{m}^*(\tilde{q}_{n+m}-q_{n+m})\Bigg\rvert +\sum_{n=-d}^d\Bigg\lvert \sum_{m=0}^d  (\tilde{q}_{n+m}-q_{n+m})(\tilde{q}_m^*-q_m^*)\Bigg\rvert 
\end{align}

The first term in \eqref{eq:PhiPQ} is seen to be zero by again appealing to the $\ell^1$--$\ell^2$ norm inequality and using the fact that $\Phi(P,Q)=0$. The remaining terms are bounded as follows. By writing
\begin{equation}
q_n=\frac{1}{2\pi\ii}\int_{\T}\Bigg( \sum_{m=0}^d q_m  z^m \Bigg)  \frac{\mathrm{d}z}{z^{n+1}}  \quad (n\in [d]_0),
\end{equation}
we see that
\begin{equation}\label{eq:qn1}
\lvert q_n\rvert \leq \lVert Q(z)\rVert_{\infty,\T}\leq 1 \quad (n\in [d]_0). 
\end{equation}
Hence, \eqref{eq:PhiPQ} with \eqref{eq:qn1} and the bound on $|q_n-\tilde q_n|$ from Theorem~\ref{thm:error} gives
\begin{equation}
    \tilde\Phi(P,\tilde{Q}_{2,N})\leq (2d+1)(d+1)\varepsilon+(2d+1)(d+1)\varepsilon+(2d+1)(d+1)\varepsilon^2< 3(2d+1)(d+1)\varepsilon,
\end{equation}
which is \eqref{eq:objbound}.

\subsection{Proof of Theorem~\ref{thm:error downscaling}}
\label{subsec:downscaling proof}

First, note that we have
\begin{equation}
\label{eq:downscaled P inequality}
    \big\lVert |P(z)|^2 - |(1-\tfrac{\varepsilon}4)P(z)|^2\big\rVert_{\infty,\T} \leq  \big\lVert P(z)+(1-\tfrac{\varepsilon}4)P(z)\big\rVert_{\infty,\T} \big\lVert P(z)-(1-\tfrac{\varepsilon}4)P(z)\big\rVert_{\infty,\T} < \frac{\varepsilon}2.
\end{equation}
Let $Q$ be the exact canonical complementary polynomial to $(1-\tfrac{\varepsilon}4)P$ in the form \eqref{eq:Qpoly}. From Theorem~\ref{thm:error}, it follows that
\begin{equation}
    |q_n-\tilde q_n|< \frac{\varepsilon}{5(d+1)} \quad (n\in [d]_0),
\end{equation}
which implies that
\begin{equation}
    \big\lVert Q(z)-\tilde Q_{2,N}(z)\big\rVert_{\infty,\T}<\frac{\varepsilon}5
\end{equation}
and
\begin{equation}\label{eq:QQprime}
    \big\lVert |Q(z)|^2 - |\tilde Q_{2,N}(z)|^2\big \rVert_{\infty,\T}
    \leq  \big\lVert 2Q(z) - Q(z) + \tilde Q_{2,N}(z)\big\rVert_{\infty,\T}\big\lVert Q(z) - \tilde Q_{2,N}(z)\big\rVert_{\infty,\T}  \le \bigg(2 + \frac{\varepsilon}5\bigg)\frac{\varepsilon}5 < \frac{\varepsilon}2.
\end{equation}

Using \eqref{eq:downscaled P inequality}, \eqref{eq:QQprime}, and the fact that $\Phi((1-\tfrac{\varepsilon}4)P,Q)=0$, we obtain
\begin{equation}
    \Phi(P,\tilde Q_{2,N}) = \big\lVert |P(z)|^2 + |\tilde Q_{2,N}(z)|^2 -1\big \rVert_{\infty,\T} < \frac{\varepsilon}2 + \frac{\varepsilon}2 + \big\lVert |(1-\tfrac{\varepsilon}4)P(z)|^2 + |Q(z)|^2 - 1\big\rVert_{\infty,\T} = \varepsilon.
\end{equation}
 The scaling \eqref{eq:algo2 complexity} can be obtained by making the replacements $\varepsilon\to\frac{\varepsilon}{5(d+1)}$ and $\delta\to\frac{\varepsilon}4$ in \eqref{eq:algo1 complexity}.

\section{Numerical results}
\label{sec:numericalresults}

In practice, our new algorithm for computation of complementary polynomials is extremely fast and accurate. A reference implementation of Algorithm~\ref{algo:main} in Python using the FFT from the PyTorch library \cite{pytorch} is provided in Figure~\ref{fig:code}.
First, we run our reference implementation for random polynomials (Section~\ref{sec:random polynomials}) and compare with prior work \cite{motlagh2024} using random polynomials. Then, we turn to practically useful examples of polynomials occurring in quantum algorithms. These include Hamiltonian simulation (Section~\ref{sec:hamiltonian simulation}), eigenvalue filtering (Section~\ref{sec:eigenvalue filtering}), and the sign function (Section~\ref{sec:sign function}). Also in these cases, our algorithm works very well; the achievable degrees are only limited by our ability to compute the approximant $P$ to the desired target function. With random polynomials we can exhibit our algorithm for far higher degrees.\\

All benchmarks are performed on an M2 Macbook Pro with 16 GB RAM; we have not used any GPU acceleration that PyTorch optionally provides for the FFT.
We compute the complementary polynomial using our algorithm for various dimensions $N$ of the FFT. To evaluate the accuracy of the output $\tilde Q_{2,N}$, we compute the loss function $\tilde\Phi(P,\tilde Q_{2,N})$ \eqref{eq:loss}, which is a measure of how well the algorithm output satisfies the complementarity condition \eqref{eq:PQ}; see Proposition~\ref{prop:norm equivalent} and Corollary~\ref{cor:error}. \\

For the random polynomials, we work in single-precision (\texttt{complex64} data type) arithmetic to facilitate comparison with the method and results presented in \cite{motlagh2024}. In the practical examples we study, we work in double-precision (\texttt{complex128} data type) arithmetic, demonstrating that Algorithm~\ref{algo:main} can achieve errors as low as $10^{-30}$.

\subsection{Random Polynomials}
\label{sec:random polynomials}

Here, we generate random polynomials $P$. The real and imaginary parts of each coefficient are independently sampled from a normal distribution of unit variance. Subsequently, we scale each polynomial to achieve $\lVert P(z)\rVert_{\infty,\T} = 1-\delta$ for choices $\delta=0.2$ and $\delta=0$. 
While our proofs in the previous section focus on $N$ even, we run the algorithm for both odd and even choices of $N$; our results carry over. \\

Figure \ref{fig:delta0.2} exhibits Algorithm~\ref{algo:main} on random polynomials with $\delta=0.2$, i.e., $\lVert P(z)\rVert_{\infty,\T} \leq 0.8$. We consider random polynomials up to degree $d=10^7$ and can numerically confirm the scaling $N = O\left(d\log\frac1{\varepsilon}\right)$ from Theorem~\ref{thm:error}, up to logarithmic terms in $d$. It appears that the bound on $N$ in Theorem~\ref{thm:error} is not tight and lower $N$ are sufficient in practice. For practical applications we therefore suggest increasing $N$ until the desired accuracy is reached, rather than using~\ref{eq:Nbound}. \\

The optimization-based method \cite{motlagh2024} uses the same loss function $\tilde\Phi$ \eqref{eq:Phi}. The best loss it achieves is indicated by a horizontal dashed line in Figure~\ref{fig:loss functions}.
Already at $N=4d$, our algorithm achieves a far better loss function. The runtime of Algorithm~\ref{algo:main} is shown in Figure~\ref{fig:runtime}. We observe a runtime $O(N\log N)$, as expected from an FFT-based algorithm. In comparison with the algorithm of \cite{motlagh2024} run on a CPU, we achieve far better runtimes for the same values of the loss function. \\

Beyond its numerical efficacy, Algorithm~\ref{algo:main} provably and reproducibly targets the same canonical solution, where $Q$ has no roots in $\mathbb{\D}$. This is in contrast to the optimization based algorithm of \cite{motlagh2024}, where the polynomial coefficients are optimized with respect to the loss function; this optimization can (and does) converge to different solutions $Q$ for the same target polynomial $P$ on successive runs. There is also no guarantee that the optimization does not converge to a local minimum. \\

For polynomials with $\delta=0$, $Q_0(z)\neq1$ because $1-|P(z)|^2$ has roots on $z\in\T$. Theorem~\ref{thm:error downscaling} proves that in those cases, Algorithm~\ref{algo:downscaling} provides a solution, scaling as $N =O\!\left(\frac{d}{\varepsilon}\log\frac{d}{\varepsilon}\right)$ for desired error $\varepsilon$. Instead of running Algorithm~\ref{algo:downscaling}, including initial downscaling of the polynomial, here we simply run Algorithm~\ref{algo:main} despite lack of rigorous justification and show the results in Figure~\ref{fig:delta0}.
We are running the same code of the reference implementation (Figure~\ref{fig:code}) used for $\delta=0.2$. Even though Algorithm~\ref{algo:main} assumes $Q_0(z)=1$, which is not the case any more, it still seems to give good results, even at degree $d=10^7$. In fact, we can still achieve losses better than those of \cite{motlagh2024} at significantly lower runtimes.
The computation of $\log(1-|P(\omega_N^n)|^2)$ in step 2 of Algorithm~\ref{algo:main} does not cause any issues, despite $1-|P(z)|^2$ having roots on $z\in\T$, because, generically, these will not be located exactly at roots of unity $\omega_N^n$ (otherwise, the polynomial can be rotated $P(z)\to P(\ee^{\ii\alpha}z)$ by a suitable phase $\ee^{\ii\alpha}$).
Empirically, by performing a fit on the data in the figure, we find a scaling of the algorithm as $N= O\big(\frac1{\sqrt[4]{\varepsilon}}\big)$ for a desired accuracy $\tilde{\varepsilon}$ in the loss function $\tilde\Phi$; proving this scaling could be the subject of future work.

\begin{figure}
\footnotesize%
\begin{verbatim}
import torch

def complementary(poly, N):
    """Algorithm 1 to compute the complementary polynomial
    Parameters:
    poly : length (d+1) vector of monomial coefficients of P(z)
    N int : size of the FFT, N >= (d+1)
    Returns:
    length (d+1) vector of monomial coefficients of Q(z)"""

    # Pad P to FFT dimension N
    paddedPoly = torch.zeros(N, dtype=torch.complex128)
    paddedPoly[:poly.shape[0]] = poly

    # Evaluate P(omega) at roots of unity omega
    pEval = torch.fft.ifft(paddedPoly, norm="forward")

    # Compute log(1-|P(omega)|^2) at roots of unity omega
    theLog = torch.log(1-torch.square(torch.abs(pEval)))

    # Apply Fourier multiplier in Fourier space
    modes = torch.fft.fft(theLog, norm="forward")
    modes[0] *= 1/2 # Note modes are ordered differently in the text
    modes[N//2+1:] = 0
    theLog = torch.fft.ifft(modes, norm="forward")    

    # Compute coefficients of Q
    coefs = torch.fft.fft(torch.exp(theLog), norm="forward")

    # Truncate to length of Q polynomial
    q = coefs[:poly.shape[0]]

    return q
\end{verbatim}
\vspace{-0.5cm}
\caption{Reference implementation in Python of Algorithm~\ref{algo:main} for finding the complementary polynomial. Input and output are (complex) coefficient vectors of $P(z)$ and $Q(z)$ in the monomial bases, along with integer $N$ controlling the accuracy of the output.}
\label{fig:code}
\end{figure}

\begin{figure}
    \begin{subfigure}[b]{0.5\textwidth}
        \centering
        \includegraphics[width=\textwidth]{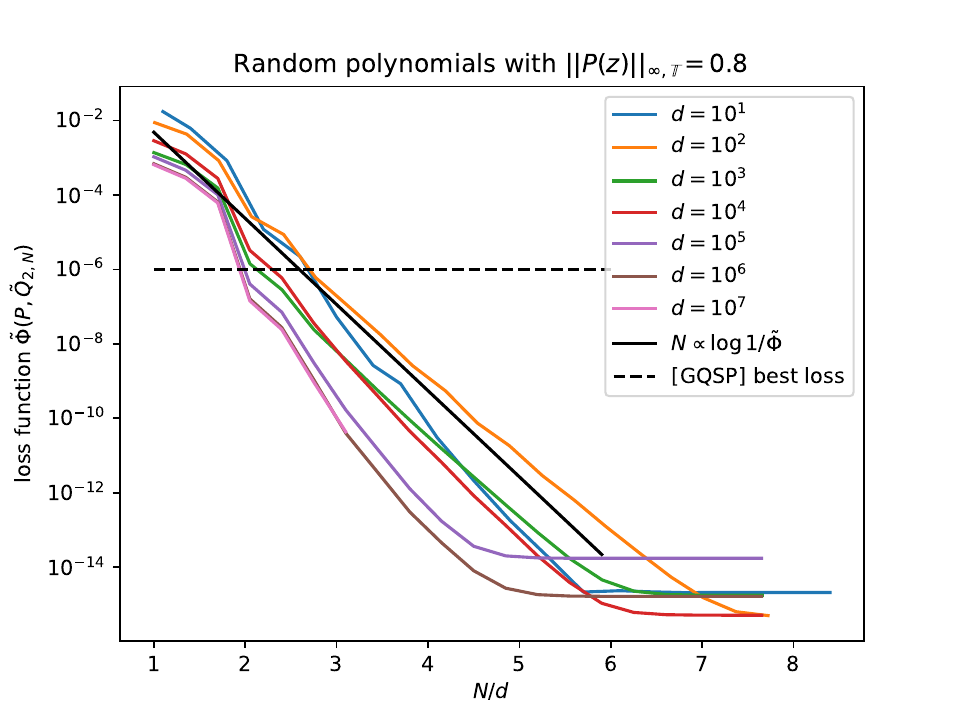}
         \caption{\begin{minipage}{0.8\textwidth}Polynomials with $\delta=0.2$. We observe the scaling $N= O\big(d \log \frac{1}{\tilde\Phi}\big)$, linear scaling with the degree $d$ and logarithmic scaling with the accuracy $\tilde\Phi$. This is commensurate with our proof of error scaling (Theorem~\ref{thm:error}) for Algorithm~\ref{algo:main}, up to logarithmic factors in $d$.\end{minipage}}
         \label{fig:delta0.2}
     \end{subfigure}
     \hfill
    \begin{subfigure}[b]{0.5\textwidth}
        \centering
        \includegraphics[width=\textwidth]{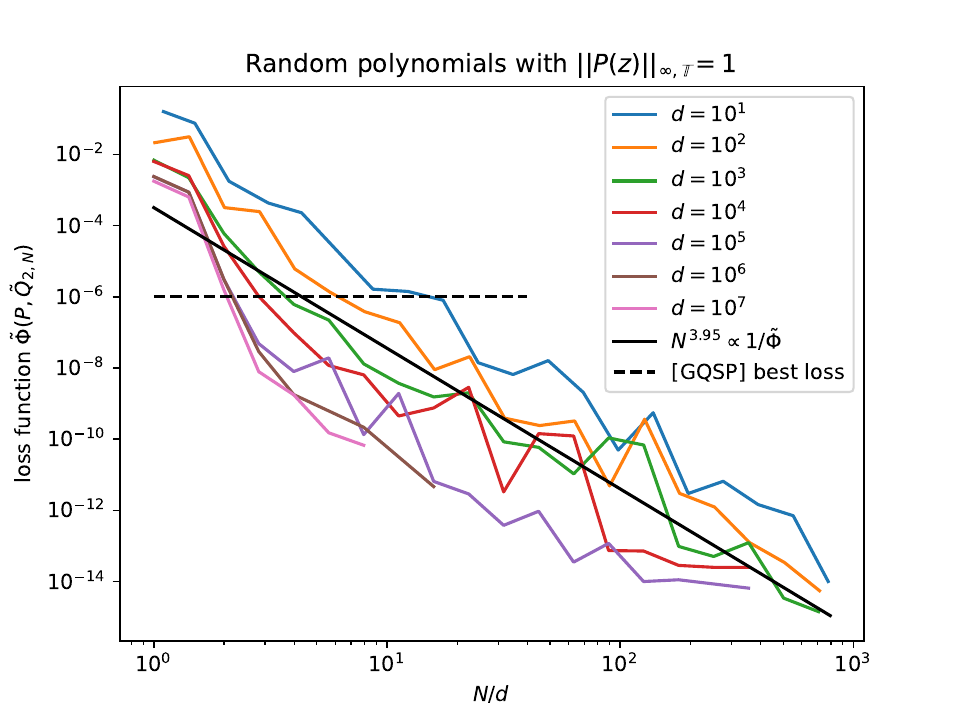}
         \caption{\begin{minipage}{0.8\textwidth}Polynomials with $\delta=0$. Despite lack of rigorous justification for Algorithm~\ref{algo:main} in this regime (necessitating the use of Algorithm~\ref{algo:downscaling}, slightly scaling down the polynomial), we observe that Algorithm~\ref{algo:main} is efficacious with a scaling behavior of $N= O\big( \frac{d}{\sqrt[4]{\tilde\Phi}}\big)$ in the accuracy $\tilde\Phi$.\end{minipage}}
         \label{fig:delta0}
    \end{subfigure}
    \caption{Finding the complementary polynomial with the reference implementation (Figure~\ref{fig:code}) of Algorithm~\ref{algo:main} for random polynomials of various degrees $d$; see Section~\ref{sec:random polynomials}. We plot the achieved loss function \eqref{eq:loss} by chosen FFT dimension $N$.
    The optimization code from GQSP \cite{motlagh2024} only reaches losses up to $10^{-6}$, which Algorithm~\ref{algo:main} achieves at far faster runtimes; see Figure~\ref{fig:runtime}. In the plots, the loss functions saturate around $10^{-14}$ due to single-precision floating-point arithmetic.}
    \label{fig:loss functions}
\end{figure}

\begin{figure}
    \begin{center}
        \includegraphics[width=0.5\textwidth]{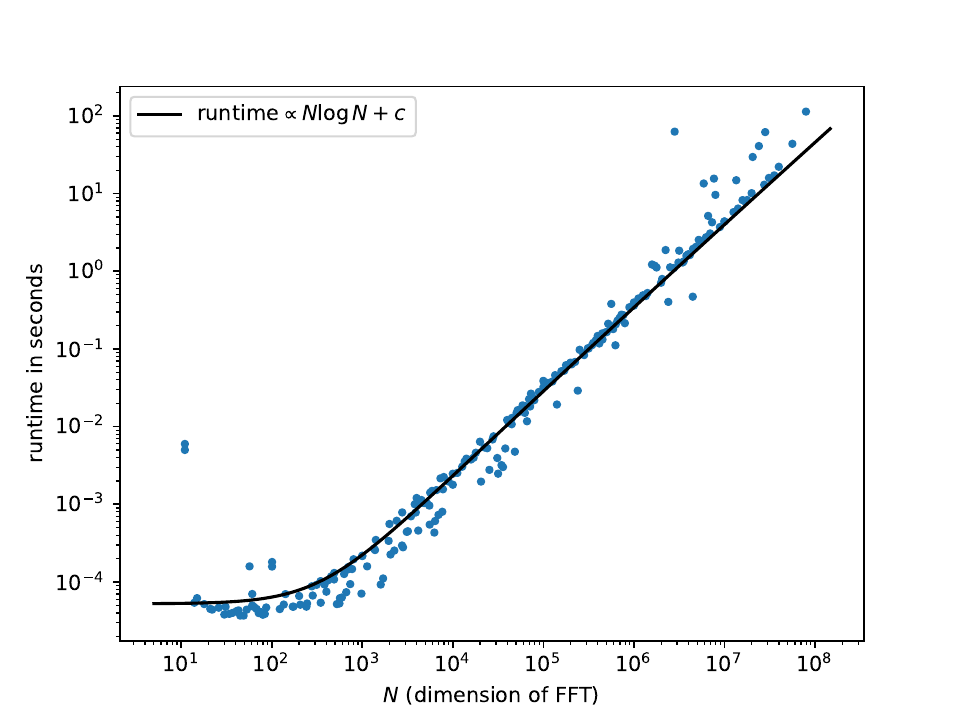}
    \end{center}
    \caption{Runtime of the reference implementation (Figure~\ref{fig:code}) of Algorithm~\ref{algo:main} on an M2 Macbook Pro with 16GB RAM for random polynomials; see Section~\ref{sec:random polynomials}. We observe the scaling $N\log N$ due to the FFTs.}
    \label{fig:runtime}
\end{figure}

\subsection{Hamiltonian simulation}
\label{sec:hamiltonian simulation}

The Jacobi-Anger expansion reads
\begin{equation}\label{eq:Jacobi-Anger}
\ee^{-\ii \tau x}=J_0(\tau)+2\sum_{n=1}^{\infty} \ii^n J_n(\tau)T_n(x) \quad (\tau\in [0,\infty), x\in [-1,1])
\end{equation}
where $\{J_n(\tau)\}_{n=0}^{\infty}$ are Bessel functions of the first kind \cite[Chapter~10]{DLMF} and $\{T_n(x)\}_{n=0}^{\infty}$ are Chebyshev polynomials of the first kind \cite[Chapter~18]{DLMF}. In quantum algorithms, the QSVT is used to apply this function to a Hamiltonian to construct the time evolution operator \cite{gilyen2019,martyn2021}. \\

Denote the truncation of the Chebsyhev series \eqref{eq:Jacobi-Anger} by
\begin{equation}
f_M(x;\tau)\coloneqq J_0(\tau)+2\sum_{n=1}^{M} \ii^n J_n(\tau)T_n(x) \quad (x\in [-1,1]).
\end{equation}
It may be shown that \cite{dong2021}
\begin{equation}
\big\lVert f_M(x;\tau)-\ee^{-\ii \tau x}\big\rVert_{\infty,[-,1,1]} < \ee^{\frac12 \ee \tau-M}.
\end{equation}
Thus, choosing $M=\big\lceil\tfrac12 \ee \tau+\log\tfrac{1}{\varepsilon} \big\rceil$, we are guaranteed that the polynomial
\begin{equation}
\tilde{f}_M(x;\tau)\coloneqq \frac{1}{1+\varepsilon}f_M(x;\tau)   
\end{equation}
$\varepsilon$-approximates $\ee^{-\ii \tau x}$ on $[-1,1]$ and satisfies $\lVert \tilde{f}_M(\tau;z)\rVert_{\infty,[-1,1]}<1$. The corresponding polynomial on $\T$ is given by $P(z)= z^M\tilde{f}_M(\tfrac12(z+z^{-1});\tau)$; see Appendix~\ref{app:QSP}. We show results of Algorithm~\ref{algo:main} in Figure~\ref{fig:hamiltonian simulation loss}, which can accurately compute the complementary polynomial with only very low overhead in FFT dimension $N$.

\begin{figure}
    \begin{center}
        \includegraphics[width=0.5\textwidth]{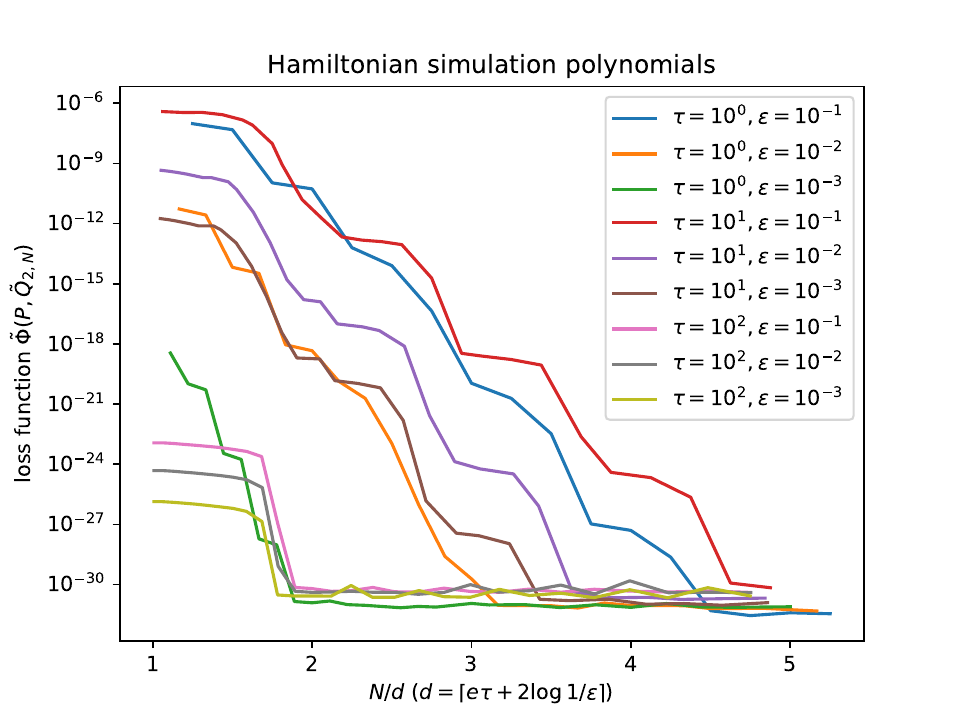}
    \end{center}
    \caption{Computing the complementary polynomial with Algorithm~\ref{algo:main} for Hamiltonian Simulation polynomials; see Section~\ref{sec:hamiltonian simulation}.}
    \label{fig:hamiltonian simulation loss}
\end{figure}

\subsection{Eigenvalue filtering}
\label{sec:eigenvalue filtering}

The polynomial defined by
\begin{equation}\label{eq:RN}
g_M(x;a)\coloneqq \frac{T_M(\frac{2x^2-(1+a^2)}{1-a^2})}{T_M(-\frac{1+a^2}{1-a^2})} \quad (a\in (0,1)).
\end{equation}
is used in quantum algorithms for eigenvalue filtering \cite{Lin_2020}. It has a sharp peak at $x=0$; see Figure~\ref{fig:eigenvalue filtering function} for an example. In quantum algorithms it can be used to project onto the kernel of a matrix. \\

\begin{figure}
    \begin{subfigure}[b]{0.5\textwidth}
        \centering
        \includegraphics[width=\textwidth]{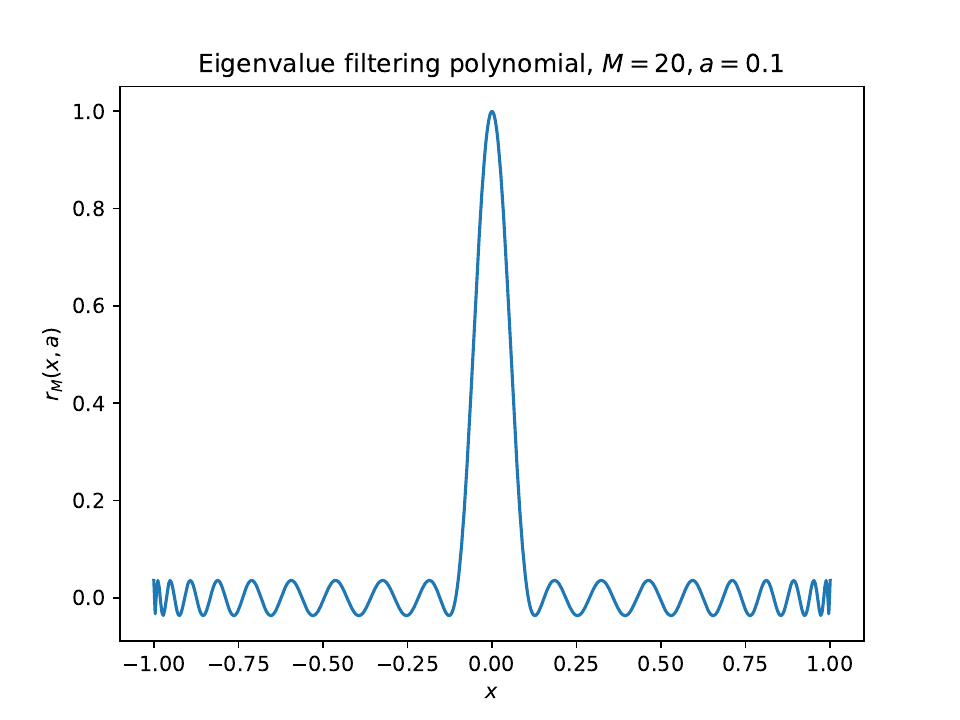}
         \caption{\begin{minipage}{0.8\textwidth}Example polynomial for eigenvalue filtering \cite{Lin_2020}.\end{minipage}}
         \label{fig:eigenvalue filtering function}
     \end{subfigure}
     \hfill
    \begin{subfigure}[b]{0.5\textwidth}
        \centering
        \includegraphics[width=\textwidth]{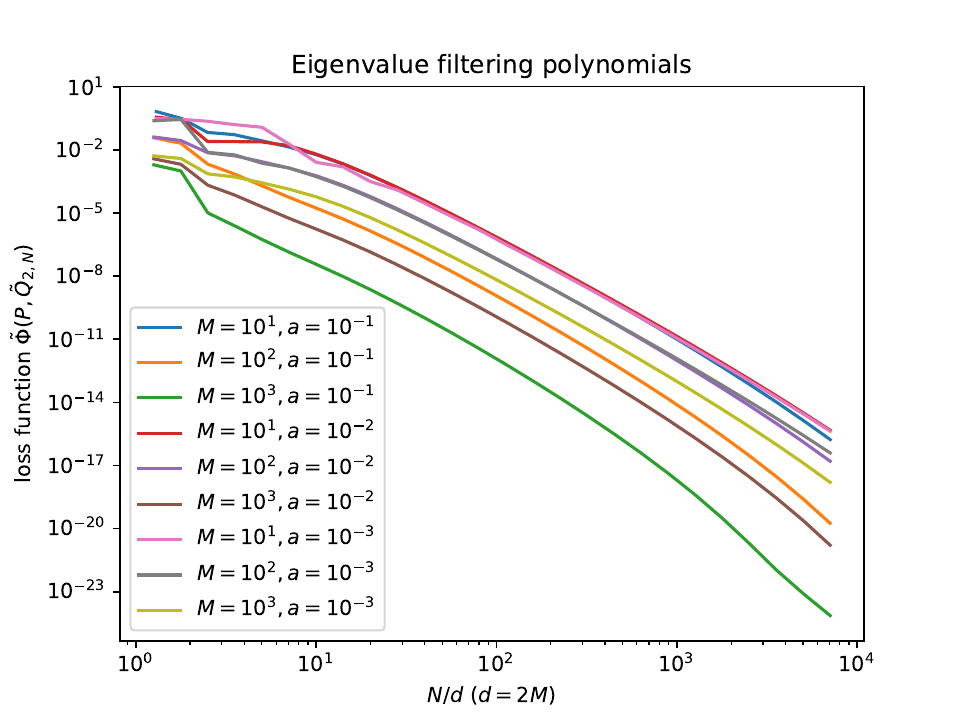}         \caption{\begin{minipage}{0.8\textwidth}Computing the complementary polynomial with Algorithm~\ref{algo:main}.\end{minipage}}
         \label{fig:eigenvalue filtering loss}
    \end{subfigure}
    \caption{Eigenvalue filtering polynomial; see Section~\ref{sec:eigenvalue filtering}.}
\end{figure}

While the expansion of \eqref{eq:RN} in the Chebyshev basis can be found explicitly via standard identities, it leads to expressions for coefficients that are numerically unstable. Instead, the Chebsyhev coefficients of \eqref{eq:RN} may be determined via a Chebyshev transform, as we now describe. Setting
\begin{equation}
x_{m,M}\coloneqq \cos\bigg(\frac{2m+1}{2M}\pi \bigg), 
\end{equation}
and using the discrete orthogonality of the Chebsyhev polynomials, it may be shown that
\begin{equation}
g_M(x;a)=c_0+\sum_{m=1}^{2M}c_m T_m(x),   
\end{equation}
where
\begin{equation}
c_0=\frac{1}{2M+1}\sum_{m=0}^{2M}g_M(x_{m,2M+1};a) ,\quad c_n=\frac{2}{2M+1}\sum_{m=0}^{2M}g_M(x_{m,2M+1};a)T_m(x_{m,2M+1}).
\end{equation}
We perform the Chebyshev transformation with the \texttt{Chebyshev.interpolate()} function in NumPy \cite{numpy}.
The corresponding polynomial on $\T$ is $P(z)=z^M g_M\big(\tfrac12\big(z^{\frac12}+z^{-\frac12}\big);a\big)$; see Appendix~\ref{app:QSP}. \\

In order to run our reference implementation of Algorithm~\ref{algo:main}, we multiply the polynomial $R_M(z;a)$ by $(1-10^{-10})$. The results in Figure~\ref{fig:eigenvalue filtering loss} demonstrate that our algorithm works very well.

\subsection{Signum function}
\label{sec:sign function}

Uniform polynomial approximations of the signum function, defined by
\begin{equation}
\mathrm{sgn}(x)\coloneqq \begin{cases}
-1 & x<0 \\
0 & x=0 \\
+1 & x> 0,
\end{cases}
\end{equation}
are required in various QSVT-based applications including amplitude amplification and phase estimation \cite{martyn2021}. As $\mathrm{sgn}(x)$ is not regular at $x=0$, a standard approach to the construction of polynomial approximants uses an error function with an appropriately rescaled argument as a regularization of the signum function. The resulting polynomial approximant is \cite{wan2022}
\begin{equation}
h_M(x;\beta)\coloneqq 2\ee^{-\beta}\sqrt{\frac{2\beta}{\pi}}\Bigg(I_0(\beta)T_1(x)+\sum_{n=1}^M (-1)^n I_n(\beta)\bigg(\frac{T_{2n+1}(x)}{2n+1}-\frac{T_{2n-1}(x)}{2n-1}\bigg)\Bigg),
\end{equation}
where $\{I_n(\beta)\}_{n=0}^{\infty}$ are modified Bessel functions \cite[Chapter~10]{DLMF}. \\

The following is a essentially a variant of \cite[Theorem~3]{wan2022}, convenient for our purposes. Below, $W_0$ denotes the principal branch of the Lambert $W$-function \cite[Section~4.13]{DLMF}.

\begin{theorem}[Chebyshev approximation of the signum function, \cite{wan2022}]
\label{thm:signum}

For any $a\in (0,1)$ and $\varepsilon\in\big(0,\frac{3}{
\sqrt{8\pi\log2}}\big)$, let $\beta>0$ satisfy $\beta\geq \frac{1}{4a^2}W_0\big(\frac{18}{\pi\varepsilon^2}\big)$ and $M\in \mathbb{Z}_{\geq 1}$ satisfy
\begin{equation}
M\geq \sqrt{\frac{W_0\big(\frac{72}{\pi\varepsilon^2}\big)\bigg(\log\bigg(\frac{3}{\sqrt{2\pi}}\frac{1}{\varepsilon\sqrt{W_0\big(\frac{72}{\pi\varepsilon^2}\big)}}\bigg)-\beta\bigg)}{W_0\bigg(\frac{1}{\ee}\bigg(\frac{1}{\beta}\log\bigg(\frac{3}{\sqrt{2\pi}}\frac{1}{\varepsilon\sqrt{W_0\big(\frac{72}{\pi\varepsilon^2}\big)}}\bigg)-1\bigg)\bigg)}}.     
\end{equation}
Then, 
\begin{equation}
\lVert \mathrm{sgn}(x)-h_M(x;\beta)\rVert_{\infty,[-1,-a]\cup [a,1]}<\varepsilon    
\end{equation}
and
\begin{equation}
\lVert h_M(x;\beta)\rVert_{\infty,[-1,1]} < 1+\tfrac23\varepsilon.
\end{equation}
hold.
\end{theorem}

Based on Theorem~\ref{thm:signum}, we consider the polynomial
\begin{equation}
\tilde{h}_M(x;\beta)\coloneqq \frac{1}{1+\tfrac23\varepsilon}h_M(x;\beta)
\label{eq:sign function final poly}
\end{equation}
which $\varepsilon$-approximates the signum function on $[-1,-a]\cup[a,1]$ and is strictly bounded in absolute value by unity on $[-1,1]$. The corresponding polynomial on $\T$ is $P(z)=z^{(2M+1)/2}\tilde{h}_M\big(\tfrac12\big(z^{\tfrac12}+z^{-\tfrac12}\big);\beta\big)$; see Appendix~\ref{app:QSP}. \\

To test our algorithm, we use example polynomials given by the parameters in Figure~\ref{fig:sign function polynomials}. Our results in Figure~\ref{fig:sign function polynomials} show that Algorithm~\ref{algo:main} works very well in practice with only small overhead in the FFT dimension $N$.

\begin{figure}
    \begin{subfigure}[b]{0.5\textwidth}
        \centering
        \begin{tabular}{r|llrr}
Example & \multicolumn{4}{c}{Specification of polynomial} \\ & $a$& $\varepsilon$& $\beta$& $M$ \\\hline
 1& $10^{-1}$ & $10^{-1}$ & 120 & 29 \\
 2& $10^{-1}$ & $10^{-4}$ & 433 & 99 \\
 3& $10^{-1}$ & $10^{-7}$ & 765 & 172 \\
 4& $10^{-1}$ & $10^{-10}$ & 1101 & 246 \\
 5& $10^{-4}$ & $10^{-1}$ & 119631742 & 26690 \\
 6& $10^{-4}$ & $10^{-4}$ & 432869078 & 89806 \\
 7& $10^{-4}$ & $10^{-7}$ & 764051835 & 156148 \\
\end{tabular}
         \caption{\begin{minipage}{0.8\textwidth}Parameters of example signum function polynomials; see Theorem~\ref{thm:signum}.\end{minipage}}
         \label{fig:sign function polynomials}
     \end{subfigure}
     \hfill
    \begin{subfigure}[b]{0.5\textwidth}
        \centering
        \includegraphics[width=\textwidth]{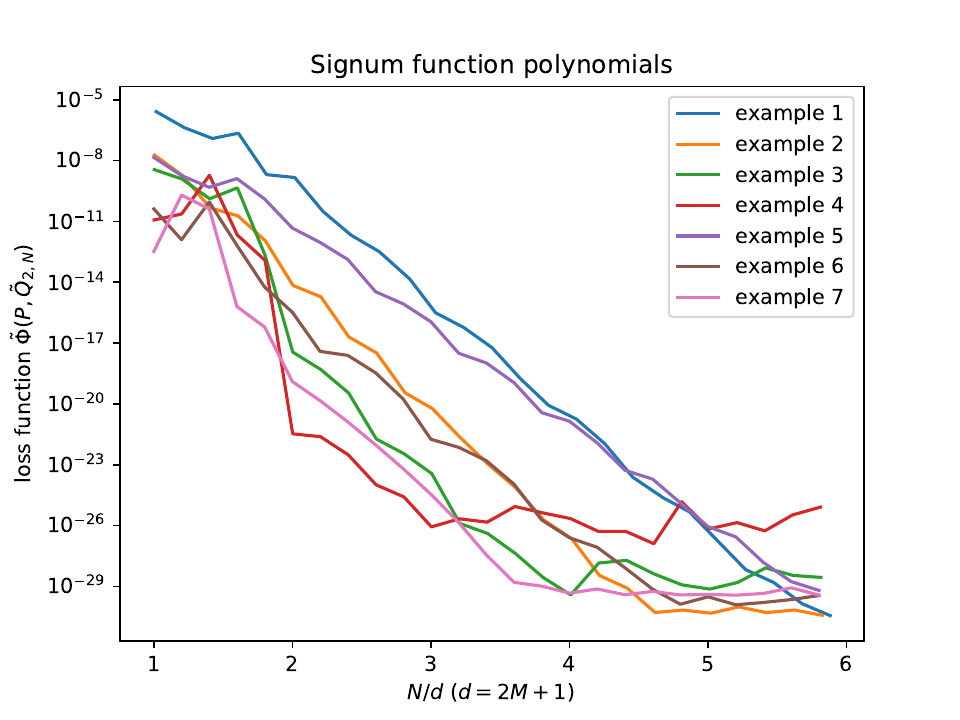}         \caption{\begin{minipage}{0.8\textwidth}Computing the complementary polynomial with Algorithm~\ref{algo:main}.\end{minipage}}
         \label{fig:sign function loss}
    \end{subfigure}
    \caption{Signum function polynomial; see Section~\ref{sec:sign function}.}
\end{figure}

\section{Discussion}
\label{sec:discussion}

In this paper, we have addressed the analytic and numerical solvability of the complementary polynomials problem, Problem~\ref{prob:CP}. 
Our main mathematical result, Theorem~\ref{thm:main}, is an exact representation, written as a set of contour integrals, for the complementary polynomial throughout the entire complex plane. We use a Fourier analytic variant of Theorem~\ref{thm:main}, Corollary~\ref{cor:main}, as a basis for developing a numerical method to obtain the complementary polynomial explicitly in the monomial basis.  \\

We give the following closing remarks. 
\begin{enumerate}

\item Problem~\ref{prob:CP} is a special case of the Fej\'{e}r-Riesz problem. Our methods for solving Problem~\ref{prob:CP} both analytically and numerically are equally applicable to the more general F\'{e}jer-Riesz problem.

\item We constructed integral representations of $Q$ on the entire complex plane. It is interesting to consider if these integrals might be explicitly computable. It has been shown \cite{mashreghi2002} that the real line Hilbert transforms of logarithms of polynomials are expressible in terms of the roots of the polynomial. As the integration in \eqref{eq:QT} amounts to a periodic Hilbert transformation of a logarithm of a Laurent polynomial, we have obtained an analogous result by comparing \eqref{eq:QT} with \eqref{eq:QFR}. As any root finding algorithm is anticipated to be more expensive than Algorithm~\ref{algo:main}, this observation is not of practical consequence. 

\item While we have chosen the representation of $Q$ on $\T$ \eqref{eq:QT} and its Fourier analytic equivalent \eqref{eq:QFourier}	as the basis for our numerical method, it is also possible to construct $Q$ though interpolation at any $d+1$ distinct points in the complex plane. Theorem~\ref{thm:main} provides a basis to do this. The representations \eqref{eq:QD} and \eqref{eq:QCD} are essentially Cauchy transforms, which the results of \cite{olver2011} show are computable $O(N\log N)$ time. However, the interpolation through $z$ values that are not roots of unity will be more expensive as the FFT is not applicable.
\item There are interesting connections between the present work and the classical signal processing literature. 
In the context of pulse synthesis, generically non-polynomial solutions to \eqref{eq:PQ} have been investigated on both $\R$ \cite{epstein2004} and $\T$ \cite{magland2005}. We particularly highlight that in \cite{magland2005}, a formula similar to \eqref{eq:QT} appears, apparently obtained by conformal mapping of an analogous formula on the real line in \cite{epstein2004}. In the context of phase retrieval, similar technologies as in this paper, namely Hilbert transforms of logarithms, have been employed \cite{nakajima1988}.

\item In \cite{motlagh2024}, the complementary polynomials problem was rephrased as an optimization problem over the coefficients of $Q$. The corresponding numerical method was based on minimizing the objective function \eqref{eq:loss}. Equation \eqref{eq:loss} defines a system of nonlinear algebraic equations for the coefficients $(q_n)_{n=0}^d$ and it is an interesting question whether this system could be solved explicitly using discrete mathematics or finite-dimensional linear algebra, without appealing to analysis as in this paper. While we make no claim this is impossible, we note that Problem~\ref{prob:CP} is essentially a special case of the Fej\'{e}r-Riesz problem, for which we are unaware of an explicit solution by such means. 

\item In the proof of Theorem~\ref{thm:main}, we have characterized the distinct solutions of Problem~\ref{prob:CP} and isolated the canonical solution having all roots outside the unit disk. The fact that the algorithm proposed in \cite{motlagh2024} does not target a particular solution may explain, in part, its effectiveness. It is an interesting question to consider the effect of the particular solution targeted on the resulting phase factors associated to a pair of complementary polynomials $(P,Q)$. In particular, does the geometry of the roots of $Q$ affect the structure of the phase factors?

\item  After this paper was posted to the arXiv, several related works have appeared.
In \cite{alexis2024}, an FFT-based algorithm for computing individual phase factors is developed. More specifically, \cite{alexis2024} presents a similar method to Algorithm~\ref{algo:main} for computing complementary polynomials as a subroutine in their algorithm, and subsequently shows how to compute individual phase factors via an FFT and linear algebra. Improvements to the linear algebraic component of the algorithm were reported in \cite{ni2024}. In \cite{skelton2024,skelton2025}, a Newton-Raphson-based algorithm for computing complementary polynomials is introduced and analyzed. Numerical experiments performed on our algorithm in \cite{skelton2024} overlap with the examples in  Sections~\ref{sec:random polynomials}--\ref{sec:hamiltonian simulation}.

\item Qualtran \cite{harrigan2024} is a recently-introduced platform for quantum algorithm development. As described in \cite{harrigan2024}, Algorithm~\ref{algo:downscaling} has been integrated into Qualtran for the purpose of constructing QSP and QSVT circuits.

\end{enumerate}

\paragraph{Acknowledgements.}

We thank Earl Campbell, Gy\"{o}rgy Geher, Tharon Holdsworth, Tanuj Khattar, Edwin Langmann, Nicholas Rubin, and especially Hari Krovi for useful discussions. We are grateful to Nina Glaser for bringing our attention to \cite{Lin_2020} and suggesting eigenvalue filtering as an example. We thank three anonymous referees for suggestions that improved the paper. This work was partially funded by Innovate UK (grant reference 10071684).

\paragraph{Data Availability.} Data sets generated during the course of this study will be made available upon reasonable request. 

\paragraph{Conflict of interest.} On behalf of all authors, the corresponding author states that there is no conflict of interest.

\appendix

\section{Quantum signal processing conventions}
\label{app:QSP}

In this appendix, we substantiate our claim that constructing complementary polynomials in different QSP conventions may be achieved by using the results of Theorem~\ref{thm:main} in conjunction with mappings between the conventions. We first show, in Section~\ref{subsec:GQSPtoLQSP}, that GQSP implies the Laurent formulation of QSP due to Haah \cite{haah2019}. In Section~\ref{subsec:LQSPtoQSP}, we show that the Laurent formulation of QSP implies the standard formulation of QSP; see \cite{martyn2021} for a similar presentation of interrelations between different parameterizations of standard QSP. 

\subsection{From GQSP to Laurent QSP}
\label{subsec:GQSPtoLQSP}

We begin by computing the remaining entries in the target matrix of GQSP.

\begin{proposition}
\label{prop:GQSP}
The precise form of the matrix realized by the GQSP sequence \eqref{eq:GQSP} is
\begin{equation}\label{eq:PQud}
\left(\begin{array}{cc}
P(z) & Q(z) \\
* & * 
\end{array}\right)= \left(\begin{array}{cc}
P(z) & Q(z) \\
u_d(z)Q^*(1/z) & -u_d(z)P^*(1/z)
\end{array}\right),
\end{equation}
where
\begin{equation}\label{eq:ud}
u_d(z)=(-1)^{d}z^d\ee^{\ii\big(\lambda+\sum_{j=0}^d \phi_j\big)} . 
\end{equation}
\end{proposition}

\begin{proof}
By unitarity, \eqref{eq:PQud} must hold for some function $u_d:\T\to \T$. Using \eqref{eq:PQ}, \eqref{eq:PQud}, and 
\begin{equation}
\det\left(\begin{array}{cc}
\ee^{\ii(\lambda+\phi_0)}\cos\theta_0 & \ee^{\ii\lambda}\sin\theta_0 \\
\ee^{\ii\theta_0}\sin\theta_0 & -\cos\theta_0
\end{array}\right)=-\ee^{\ii(\lambda+\phi_0)},\quad \det\left(\begin{array}{cc}
z & 0 \\
0 & 1
\end{array}\right)=z,\quad 
\det\left(\begin{array}{cc}
\ee^{\ii\phi_j} \cos\theta_j & \sin\theta_j  \\
\ee^{\ii\phi_j}\sin\theta_j & -\cos\theta_j
\end{array}\right)=-\ee^{\ii\phi_j} 
\end{equation}
to evaluate the determinant of both sides of \eqref{eq:GQSP} gives \eqref{eq:ud}.
\end{proof}

By combining Theorem~\ref{thm:GQSP} and Proposition~\ref{prop:GQSP} and making the replacement $z\to z^2$, we have
\begin{multline}\label{eq:PQz2}
\left(\begin{array}{cc}
P(z^2) & Q(z^2) \\
u_d(z^2)Q^*(1/z^2) & -u_d(z^2)P^*(1/z^2)
\end{array}\right)=\\
\left(\begin{array}{cc}
\ee^{\ii(\lambda+\phi_0)}\cos \theta_0 & \ee^{\ii\lambda}\sin \theta_0 \\
\ee^{\ii\phi_0}\sin \theta_0 & -\cos \theta_0
\end{array}\right)\Bigg[\prod_{j=1}^d \left(\begin{array}{cc} z^2 & 0 \\ 0 & 1 \end{array}\right)\left(\begin{array}{cc}
\ee^{\ii\phi_j}\cos \theta_j & \sin \theta_j \\
\ee^{\ii\phi_j}\sin \theta_j & -\cos \theta_j
\end{array}\right) \Bigg] \quad (z\in \T).
\end{multline}
We rename parameters $\lambda\to \lambda_0$ and $\phi_j\to \phi_j+\lambda_{j+1}$ for $j\in [d-1]_0$ in \eqref{eq:PQz2}. This yields
\begin{multline}\label{eq:PQz2,2}
\left(\begin{array}{cc}
P(z^2) & Q(z^2) \\
u_d(z^2)Q^*(1/z^2) & -u_d(z^2)P^*(1/z^2)
\end{array}\right)=\\
\left(\begin{array}{cc}
\ee^{\ii(\lambda_0+\phi_0+\lambda_1)}\cos \theta_0 & \ee^{\ii\lambda_0}\sin \theta_0 \\
\ee^{\ii\phi_0+\lambda_1}\sin \theta_0 & -\cos \theta_0
\end{array}\right)\Bigg[\prod_{j=1}^{d-1} \left(\begin{array}{cc} z^2 & 0 \\ 0 & 1 \end{array}\right)\left(\begin{array}{cc}
\ee^{\ii(\phi_j+\lambda_{j+1})}\cos \theta_j & \sin \theta_j \\
\ee^{\ii(\phi_j+\lambda_{j+1})}\sin \theta_j & -\cos \theta_j
\end{array}\right) \Bigg] \\
\times \left(\begin{array}{cc} z^2 & 0 \\ 0 & 1 \end{array}\right)\left(\begin{array}{cc}
\ee^{\ii\phi_d}\cos \theta_j & \sin \theta_j \\
\ee^{\ii\phi_d}\sin \theta_j & -\cos \theta_j
\end{array}\right).
\end{multline}
Using the factorizations
\begin{equation}
\begin{split}
    \left(\begin{array}{cc}
\ee^{\ii(\lambda_0+\phi_0+\lambda_1)}\cos \theta_0 & \ee^{\ii\lambda_0}\sin \theta_0 \\
\ee^{\ii(\phi_0+\lambda_1)}\sin \theta_0 & -\cos \theta_0
\end{array}\right)=&\;     \left(\begin{array}{cc}
\ee^{\ii(\lambda_0+\phi_0)}\cos \theta_0 & \ee^{\ii\lambda_0}\sin \theta_0 \\
\ee^{\ii\phi_0}\sin \theta_0 & -\cos \theta_0
\end{array}\right)\left(\begin{array}{cc} 
\ee^{\ii\lambda_1} & 0 \\
0 & 1
\end{array}\right), \\
    \left(\begin{array}{cc}
\ee^{\ii(\phi_j+\lambda_{j+1})}\cos \theta_j & \sin \theta_j \\
\ee^{\ii(\phi_j+\lambda_{j+1})}\sin \theta_j & -\cos \theta_j
\end{array}\right)=&\;     \left(\begin{array}{cc}
\ee^{\ii\phi_j}\cos \theta_j & \sin \theta_j \\
\ee^{\ii\phi_j}\sin \theta_j & -\cos \theta_j
\end{array}\right)\left(\begin{array}{cc} 
\ee^{\ii\lambda_{j+1}} & 0 \\
0 & 1
\end{array}\right)
\end{split}
\end{equation}
and commutativity of diagonal matrices in \eqref{eq:PQz2}, we deduce
\begin{multline}\label{eq:PQz2,3}
\left(\begin{array}{cc}
P(z^2) & Q(z^2) \\
u_d(z^2)Q^*(1/z^2) & -u_d(z^2)P^*(1/z^2)
\end{array}\right)=\\
\left(\begin{array}{cc}
\ee^{\ii(\lambda_0+\phi_0)}\cos \theta_0 & \ee^{\ii\lambda_0}\sin \theta_0 \\
\ee^{\ii\phi_0}\sin \theta_0 & -\cos \theta_0
\end{array}\right)\Bigg[\prod_{j=1}^{d} \left(\begin{array}{cc} z^2 & 0 \\ 0 & 1 \end{array}\right)\left(\begin{array}{cc}
\ee^{\ii(\lambda_j+\phi_j)}\cos \theta_j & \ee^{\ii\lambda_j}\sin \theta_j \\
\ee^{\ii\phi_j}\sin \theta_j & -\cos \theta_j
\end{array}\right) \Bigg].
\end{multline}

We make the following specializations in \eqref{eq:ud} and \eqref{eq:PQz2,3},
\begin{equation}
\lambda_j=\frac{\pi}{2} \quad \phi_j=\frac{\pi}{2} ,\quad \theta_j\to \theta_j+\pi  \quad (j\in [d]_0)
\end{equation}
and multiply \eqref{eq:PQz2,2} by $z^{-d}$ to obtain
\begin{equation}\label{eq:PQz2,4}
\left(\begin{array}{cc}
z^{-d}P(z^2) & z^{-d}Q(z^2) \\
-z^d Q^*(1/z^2) & z^d P^*(1/z^2)
\end{array}\right)=\\
\left(\begin{array}{cc}
\cos \theta_0 & \ii\sin \theta_0 \\
\ii\sin \theta_0 & \cos \theta_0
\end{array}\right)\Bigg[\prod_{j=1}^{d} \left(\begin{array}{cc} z & 0 \\ 0 & z^{-1} \end{array}\right)\left(\begin{array}{cc}
\cos \theta_j & \ii\sin \theta_j \\
\ii\sin \theta_j & \cos \theta_j
\end{array}\right) \Bigg].
\end{equation}

To proceed, we recall Corollary~\ref{cor:real} and the $\mathrm{U}(1)$ invariance of Problem~\ref{prob:CP}. Together, these guarantee that given $P\in \R[z]$ satisfying the conditions of Theorem~\ref{thm:main}, there exists a $Q\in \ii \R[z]$ that is complementary to $P$. By setting $F(z)=z^{-d}P(z^2)$ and $\ii G(z)= z^{-d} Q(z^2)$ in \eqref{eq:PQz2,4}, we deduce the Laurent formulation of QSP: 

\begin{theorem}[Laurent quantum signal processing, \cite{haah2019}] \label{thm:LaurentQSP}
Let $F\in \R[z,z^{-1}]$ with $\deg F=d\in \Z_{\geq 1}$ and parity $d\bmod 2$. Then, there exists $G\in \R[z,z^{-1}]$ and $(\theta_j)_{j=0}^d\in(-\pi,\pi]^{d+1}$ such that
\begin{equation}
\lvert F(z)\rvert^2+\lvert G(z)\rvert^2=1 \quad (z\in \T)  
\end{equation}
and
\begin{equation}\label{eq:LaurentQSP}
\left(\begin{array}{cc}
F(z) & \ii G(z) \\
\ii G(z^{-1}) & F(z^{-1})
\end{array}\right)=\left(\begin{array}{cc}
\cos\theta_0 & \ii \sin\theta_0 \\
\ii\sin\theta_0 & \cos{}\theta_0
\end{array}\right)\Bigg[\prod_{j=1}^d \left(\begin{array}{cc} z & 0 \\
0 & z^{-1}\end{array}\right)\left(\begin{array}{cc}
\cos\theta_j & \ii \sin\theta_j \\
\ii\sin\theta_j & \cos{}\theta_j
\end{array}\right) \Bigg] \quad (z\in \T)
\end{equation}
hold.
\end{theorem}

To apply the results of this paper to determine $G(z)$, set $P(z)=z^{\frac{d}{2}} F(\sqrt{z})$; this polynomial satisfies the conditions of Problem~\ref{prob:CP}. If $Q$ is the canonical complementary polynomial to $P$ with purely imaginary coefficients, we see that $\ii G(z)=z^{-d}Q(z^2)$.

\subsection{From Laurent QSP to real QSP}
\label{subsec:LQSPtoQSP}

To recover the formulation of standard QSP based on polynomials on $[-1,1]$, we will employ a definition of the Chebyshev polynomials of the first kind,
\begin{equation}\label{eq:TU}
T_n(x)\coloneqq\frac12(z^n+z^{-n}) \quad (z\in \T),
\end{equation}
where $x\coloneqq \mathrm{Re}\, z$. \\

Suppose that $F$ satisfies $F(z^{-1})=F(z)$, i.e.,
\begin{equation}\label{eq:Fpalindromic}
F(z)=f_0+\frac12\sum_{n=1}^d f_n(z^n+z^{-n})= \sum_{n=0}^d f_n T_n(x),
\end{equation}
 where we have used \eqref{eq:TU}. Thus, the $(1,1)$ entry of \eqref{eq:LaurentQSP} can be understood as a real-valued polynomial on $[-1,1]$. This observation in conjunction with Theorem~\ref{thm:LaurentQSP} gives the following result; see \cite{gilyen2019} for a similar formulation without reference to complex variables. 

\begin{theorem}[Real quantum signal processing]\label{thm:realQSP}
Let $p\in \R[x]$ with $\deg p=d\in \Z_{\geq 1}$ with parity $d\bmod 2$ satisfy $\lvert p(x)\rvert\leq 1$ on $[-1,1]$. Then, there exists $G\in \R[z,z^{-1}]$ such that
\begin{equation}
\lvert p(x)\rvert^2+\lvert G(z)\rvert^2=1 \quad (z\in \T,\,x=\mathrm{Re}\,z)    
\end{equation}
and $(\theta_j)_{j=0}^d\in(-\pi,\pi]^{d+1}$ such that
\begin{equation}\label{eq:LaurentQSP2}
\left(\begin{array}{cc}
p(x) & \ii G(z) \\
\ii G(z^{-1}) & p(x)
\end{array}\right)=\left(\begin{array}{cc}
\cos\theta_0 & \ii \sin\theta_0 \\
\ii\sin\theta_0 & \cos{}\theta_0
\end{array}\right)\Bigg[\prod_{j=1}^d \left(\begin{array}{cc} z & 0 \\
0 & z^{-1}\end{array}\right)\left(\begin{array}{cc}
\cos\theta_j & \ii \sin\theta_j \\
\ii\sin\theta_j & \cos{}\theta_j
\end{array}\right) \Bigg] \quad (z\in \T, \,x=\mathrm{Re}\,z)
\end{equation}
holds.
\end{theorem}

Given $p\in \R[x]$ satisfying the conditions of Theorem~\ref{thm:realQSP}, compute the Chebyshev coefficients of $p$ and denote them by $(f_n)_{n=0}^d$. This defines a Laurent polynomial $F$ via \eqref{eq:Fpalindromic}. The recipe for constructing $G$ below Theorem~\ref{thm:LaurentQSP} is now applicable.

\section{Complex analysis}
\label{app:complex}

In this appendix, we collect classical complex analysis results that we use in the main text. For a comprehensive introduction to complex analysis, we refer to \cite{ahlfors}. Here, assuming a familiarity with basic complex analysis concepts and results, we precisely state and elaborate on the theorems employed in the main text. 

\subsection{F\'{e}jer-Riesz theorem}

The Fej\'{e}r-Riesz theorem states that a Laurent polynomial that is real and non-negative on $\T$ can be written as the squared modulus of some polynomial on $\T$. The precise statement is as follows \cite{riesz2012}. 

\begin{theorem}[Fej\'{e}r-Riesz]
Suppose that $F\in \C[z,z^{-1}]$ satisfies $F(\T)\subset \R$ and 
\begin{equation}
F(z)\geq 0 \quad (z\in \T).    
\end{equation}
Then, there exists $f\in \C[z]$ satisfying
\begin{equation}
\lvert f(z)\rvert^2 =F(z) \quad (z\in \T).    
\end{equation}
Moreover, $f$ may be chosen such that it has no roots in $\D$.
\end{theorem}

We use the Fej\'{e}r-Riesz theorem to obtain the canonical factorization \eqref{eq:Q} of $Q$ within the proof of Theorem~\ref{thm:main}.

\subsection{Schwarz integral formula}

The Schwarz integral provides a representation of a holomorphic function on the closed unit disk in terms of the values of the real part of the function on $\T$ and the value of the imaginary part at $z=0$ \cite{ahlfors}.

\begin{theorem}[Schwarz]
Suppose $f$ is holomorphic on $\overline{\D}$. Then,
\begin{equation}\label{eq:Schwarz}
f(z)=\frac{1}{2\pi\ii}\int_{\T} \frac{z'+z}{z'-z}\,\mathrm{Re}\,f(z')\frac{\mathrm{d}z'}{z'}+\ii\,\mathrm{Im}\,f(0) \quad (z\in \D).
\end{equation}
\end{theorem}

We use the Schwarz integral formula within the proof of Theorem~\ref{thm:main} to obtain the integral representation \eqref{eq:USchwarz}, which leads to \eqref{eq:QD}.

\subsection{Plemelj formula}

For a detailed introduction to Cauchy-type integrals, we refer to \cite{ablowitz2003,olver2011}. Here, we provide a simple but precise treatment of such integrals required in the main text. \\

Let $\Gamma$ be an oriented curve and $F$ a function $\Gamma\to \C$. Suppose that on $\Gamma$, $F$ has a single singularity at $z_0\in \Gamma$. Then, where it exists, the Cauchy principal value integral of $F$ on $\Gamma$ is defined by
\begin{equation}\label{eq:PV}
\pvint_{\Gamma} F(z)\,\mathrm{d}z\coloneqq \lim_{\varepsilon\downarrow 0}\int_{\Gamma\setminus B(z_0;\varepsilon)} F(z)\,\mathrm{d}z,
\end{equation}
where $B(z_0;\varepsilon)\coloneqq \{z\in \C: \lvert z-z_0\rvert< \varepsilon\}$. Our main interest in Cauchy principal value integrals stems from their appearance in the Plemelj formula.

\begin{theorem}[Plemelj]
Suppose that $\Gamma$ is a positively-oriented Jordan curve and $f$ is a continuous function $\Gamma\to \C$. Let
\begin{equation}\label{eq:Plemelj}
\mathcal{C}_{\Gamma}[f(z)]\coloneqq \frac{1}{2\pi\ii} \int_{\Gamma} \frac{f(z')}{z'-z}\,\mathrm{d}z'.
\end{equation}
Then, 
\begin{equation}
\lim_{\substack{z\to z_0 \\ z\in \Omega_{\pm}}} \mathcal{C}_{\Gamma}[f(z)]=\pm\frac{1}{2}f(z_0)+\frac{1}{2\pi\ii}\,\pvint_{\Gamma} \frac{f(z')}{z'-z_0}\,\mathrm{d}z', 
\end{equation}
where $\Omega_+$ is the domain enclosed by $\Gamma$ and $\Omega_-\coloneqq \C\setminus \overline{\Omega}_+$.
\end{theorem}

We use the Plemelj formula in the proof of Theorem~\ref{thm:main} to sequentially construct \eqref{eq:QT} and \eqref{eq:QCD} starting from \eqref{eq:QD}. Additionally, the Cauchy principal value integral is needed to define the periodic Hilbert transform \eqref{eq:H} and hence prove Corollary~\ref{cor:main}.

\bibliographystyle{unsrt}
\bibliography{CP.bib}

\end{document}